\documentclass{llncs}
\usepackage{amssymb}
\usepackage{amsmath}
\usepackage{algorithmic}
\usepackage[ruled,vlined,linesnumbered]{algorithm2e}
\usepackage{nicefrac}
\usepackage{comment}
\usepackage{color}
\usepackage{ifpdf}
\usepackage{caption}
\usepackage{subcaption}
\usepackage{graphicx}
\usepackage[bookmarks=true, pdfpagelabels, plainpages=false]{hyperref}
\usepackage{multirow}
\usepackage{tabu}

\begin{document}

\newcommand {\ignore} [1] {}

\def \aa   {\alpha}
\def \bb   {\beta}
\def \gg   {\gamma}
\def \ee   {\varepsilon}
\def \el   {\ell}
\def \ss   {\sigma}
\def \dd   {\delta}
\def \Om   {\Omega}

\def \PP   {{\cal P}}
\def \QQ   {{\cal Q}}
\def \DD   {{\cal D}}
\def \NN   {{\cal N}}
\def \AA   {{\cal A}}
\def \MM   {{\cal M}}
\def \II   {{\cal I}}
\def \TT   {{\cal T}}
\def \RR   {{\cal R}}

\newcommand{\ie}{{\it i.e.}}
\newcommand{\eg}{{\it e.g.}}
\newcommand{\bu}{{\mathbf{u}}}
\newcommand{\bc}{{\mathbf{c}}}
\newcommand{\bx}{{\mathbf{x}}}
\newcommand{\by}{{\mathbf{y}}}
\newcommand{\bv}{{\mathbf{v}}}
\newcommand{\be}{{\mathbf{e}}}
\newcommand{\bt}{{\mathbf{t}}}
\newcommand{\bD}{{\mathbf{D}}}
\newcommand{\MWC}{{\texttt{Multiway Cut}}}
\newcommand{\NMW}{{\texttt{Node Multiway Cut}}}
\newcommand{\DMW}{{\texttt{Directed Multiway Cut}}}
\newcommand{\EXT}{{\texttt{$0$-Extension}}}
\newcommand{\VC}{{\texttt{Vertex Cover}}}
\newcommand{\ML}{{\texttt{Metric Labeling}}}
\newcommand{\UML}{{\texttt{Uniform Metric Labeling}}}
\newcommand{\MS}{{\texttt{Max SAT}}}
\newcommand{\RSet}{{\mathtt{R}}}
\newcommand{\opt}{{\mathtt{opt}}}
\newcommand{\argmin}{{\text{argmin}}}
\newcommand{\Cmax}{{c_{\text{max}}}}

\newcommand{\CC}{{\textsf{Correlation Clustering}}}
\newcommand{\MMA}{{\textsf{Max Min Agreements}}}
\newcommand{\MMD}{{\textsf{Min Max Disagreements}}}
\newcommand{\MLA}{{\textsf{Max Local Agreements}}}
\newcommand{\MLD}{{\textsf{Min Local Disagreements}}}
\newcommand{\MMST}{{\textsf{Min Max $s-t$ Cut}}}
\newcommand{\MMMW}{{\textsf{Min Max Multiway Cut}}}
\newcommand{\MMM}{{\textsf{Min Max Multicut}}}
\newcommand{\MC}{{\textsf{Max Cut}}}
\newcommand{\MST}{{\textsf{Min $s-t$ Cut}}}
\newcommand{\MW}{{\textsf{Multiway Cut}}}
\newcommand{\MultiC}{{\textsf{Multicut}}}

\newcommand{\RN}[1]{%
  \textup{\expandafter{\romannumeral#1}}%
}

\def\bx {{\bf x}}
\def\ba {{\bf a}}
\def\bb {{\bf b}}
\def\bd {{\bf d}}

\pagenumbering{arabic}

\makeatletter
\renewcommand*{\@fnsymbol}[1]{\ensuremath{\ifcase#1\or \dagger\or \ddagger\or \ddagger\or
    \mathsection\or \mathparagraph\or \|\or **\or \dagger\dagger
    \or \ddagger\ddagger \else\@ctrerr\fi}}
\makeatother

\title{Local Guarantees in Graph Cuts and Clustering}

\author{Moses Charikar\inst{1} \fnmsep  \thanks{Supported by NSF grants CCF-1617577, CCF-1302518 and a Simons
Investigator Award},  Neha Gupta\inst{1} \fnmsep $^\dagger$   \and Roy Schwartz\inst{2} \fnmsep \thanks{Supported by ISF grant 1336/16} }

\institute{Stanford University, Stanford CA 94305, USA,\\
\email{\{moses,nehagupta\}@cs.stanford.edu},
\and
Technion, Haifa, 3200003, Israel,\\
\email{schwartz@cs.technion.ac.il}
}

\maketitle

\begin{abstract}
{\CC} is an elegant model that captures fundamental graph cut problems such as {\MST}, {\MW}, and {\MultiC}, extensively studied in combinatorial optimization.
Here, we are given a graph with edges labeled $+$ or $-$ and the goal is to produce a clustering that agrees with the labels as much as possible: 
$+$ edges within clusters and $-$ edges across clusters.
The classical approach towards {\CC} (and other graph cut problems) is to optimize a global objective.
We depart from this and study local objectives: minimizing the maximum number of disagreements for edges incident on a single node, and the analogous max min agreements objective.
This naturally gives rise to a family of basic min-max graph cut problems. 
A prototypical representative is {\MMST}: find an $s-t$ cut minimizing the largest number of cut edges incident on any node.
We present the following results: $(1)$ an $O(\sqrt{n})$-approximation for the problem of minimizing the maximum total weight of disagreement edges incident on any node (thus providing the first known approximation for the above family of min-max graph cut problems), $(2)$ a remarkably simple $7$-approximation for minimizing local disagreements in complete graphs (improving upon the previous best known approximation of $48$), and $(3)$ a $\nicefrac[]{1}{(2+\varepsilon)}$-approximation for maximizing the minimum total weight of agreement edges incident on any node, hence improving upon the $\nicefrac[]{1}{(4+\varepsilon)}$-approximation that follows from the study of approximate pure Nash equilibria in cut and party affiliation games.

\vspace{5pt}
\end{abstract}

\thispagestyle{empty}

\section{Introduction}\label{sec:Introduction}
Graph cuts are extensively studied in combinatorial optimization, including fundamental problems such as {\MST}, {\MW}, and {\MultiC}.
Typically, given an undirected graph $G=(V,E)$ equipped with non-negative edge weights $c:E\rightarrow \mathcal{R}_+$ the goal is to find a {\em constrained} partition $\mathcal{S}=\left\{ S_1,\ldots, S_{\ell}\right\}$ of $V$ minimizing the total weight of edges crossing between different clusters of $\mathcal{S}$.
e.g., in {\MST}, $\mathcal{S}$ has two clusters, one containing $s$ and the other containing $t$.
Similarly, in {\MW}, $\mathcal{S}$ consists of $k$ clusters each containing exactly one of $k$ given special vertices $t_1,\ldots,t_k$.
In {\MultiC}, the clusters of $\mathcal{S}$ must separate $k$ given pairs of special vertices $\left\{ s_i,t_i\right\} _{i=1}^k$.

The elegant model of {\CC} captures all of the above fundamental graph cut problems, and was first introduced by Bansal {\em et al.} \cite{bansal2004correlation} more than a decade ago.
In {\CC}, 
we are given an  undirected graph $G=(V,E)$ equipped with non-negative edge weights $c:E\rightarrow \mathcal{R} _+$.
Additionally, $E$ is partitioned into $E^+$ and $E^-$, where edges in $E^+$ ($E^-$) are considered to be labeled as $+$ ($-$).
The goal is to find a partition of $V$ into an {\em arbitrary} number of clusters $\mathcal{S}=\left\{ S_1,\ldots,S_{\ell}\right\}$ that agrees with the edges' labeling as much as possible:
the endpoints of $+$ edges are supposed to be placed in the same cluster and endpoints of $-$ edges in different clusters.
Typically, the objective is to find a clustering that minimizes the total weight of misclassified edges. 
This models, \eg, {\MST}, since one can label all edges in $G$ with $+$, and add $(s,t)$ to $E$ with a label of $-$ and set its weight to $c_{s,t}=\infty$ ({\MW} and {\MultiC} are modeled in a similar manner).

{\CC} has been studied extensively for more than a decade \cite{ailon2012improved,ailon2008aggregating,charikar2003clustering,chawla2015near,demaine2006correlation,Wirth10}.
In addition to the simplicity and elegance of the model, its study is also motivated by a wide range of practical applications:
image segmentation \cite{Wirth10}, clustering gene expression patterns \cite{amit2004bicluster,ben1999clustering}, cross-lingual link detection \cite{van2007correlation}, and the aggregation of inconsistent information \cite{filkov2004integrating}, to name a few (refer to the survey \cite{Wirth10} and the references therein for additional details).

Departing from the classical global objective approach towards {\CC}, we consider a broader class of objectives that allow us to bound the number of misclassified edges incident on any node (or alternatively edges classified correctly).
We refer to this class as {\CC} with {\em local guarantees}.
First introduced by Puleo and Milenkovic \cite{puleo2016correlation}, {\CC} with local guarantees naturally arises in settings such as community detection without antagonists, \ie, objects that are inconsistent with large parts of their community,
and has found applications in diverse areas, \eg, recommender systems, bioinformatics, and social sciences \cite{cheng2000biclustering,kriegel2009clustering,puleo2016correlation,symeonidis2006nearest}.



\vspace{5pt}
\noindent {\bf{Local Minimization of Disagreements and Graph Cuts:}}
A prototypical example when considering minimization of disagreements with local guarantees is the {\MMD} problem, whose goal is to find a clustering that minimizes the maximum total weight of misclassified edges incident on any node.

Formally, given a partition $\mathcal{S}=\left\{S_1,\ldots,S_{\ell}\right\}$ of $V$,
for $u\in S_i$, define:
$$\text{disagree}_{\mathcal{S}}(u)\triangleq \sum _{v\notin S_i:(u,v)\in E^+}c_{u,v} + \sum _{v\in S_i:(u,v)\in E^-}c_{u,v}~.$$
The objective of {\MMD} is: $\min _{\mathcal{S}}\max _{u\in V}\left\{ \text{disagree}_{\mathcal{S}}(u)\right\}$.
This is NP-hard even on complete unweighted graphs and approximations are known for only a few special cases \cite{puleo2016correlation}.
No approximation is known for general graphs.

Just as minimization of total disagreements in {\CC} models fundamental graph cut problems, {\MMD} gives rise to a variety of basic min-max graph cut problems.
A natural problem here is {\MMST}:
Its input is identical to that of {\MST}, however its objective is to find an $s-t$ cut $(S,\overline{S})$ minimizing the total weight of cut edges incident on any node: $ \min _{S\subseteq V:s\in S, t\notin S}\max _{u\in V}\{\sum _{v:(u,v)\in \delta(S)}c_{u,v}\}$.\footnote{$\delta(S)$ denotes the collection of edges crossing the cut $(S,\overline{S})$.}
Despite the fact that {\MMST} is a natural graph cut problem, no approximation is known for it.
{\MMD} also gives rise to {\MMMW} and {\MMM}, defined similarly; 
no approximation is known for these. 
One of our goals is to highlight this family of min-max graph cut problems which we believe deserve further study.
Other graph cut problems were studied from the min-max perspective, \eg, \cite{bansal2014min,svitkina2004min}.
However, the goal there is to find a constrained partition that minimizes the total weight of cut edges incident on any {\em cluster} (as opposed to incident on any {\em node}).


{\MMD} is a special case of the more general {\MLD} problem.
Given a clustering $\mathcal{S}$, consider the vector of all disagreement values $\text{disagree}_{\mathcal{S}}(V)\in \mathcal{R}_+^V$, where $(\text{disagree}_{\mathcal{S}}(V)) _u = \text{disagree}_{\mathcal{S}}(u) $ $\forall u\in V$. 
The objective of {\MLD} is to find a partition $\mathcal{S}$ that minimizes $f( \text{disagree}_{\mathcal{S}}(V))$ for a given function $f$.
For example, if $f$ is the $\max$ function {\MLD} reduces to {\MMD}, and if $f$ is the summation function {\MLD} reduces to the classic objective of minimizing total disagreements.


\vspace{5pt}
\noindent {\bf{Local Maximization of Agreements:}}
Another natural objective of {\CC} is that of maximizing the total weight of edges correctly classified \cite{bansal2004correlation,swamy2004correlation}.
A prototypical example for local guarantees is {\MMA}, 
\ie \text{ }  finding
a clustering that maximizes the minimum total weight of correctly classified edges incident on any node.
Formally, given a partition $\mathcal{S}=\{S_1,\ldots,S_{\ell}\}$ of $V$,
for $u \in S_i$, define:
$$\text{agree}_{\mathcal{S}}(u)\triangleq \sum _{v\in S_i:(u,v)\in E^+}c_{u,v} + \sum _{v\notin S_i:(u,v)\in E^-}c_{u,v}~.$$
The objective of {\MMA} is:  $ \max _{\mathcal{S}}\min _{u\in V}\{ \text{agree}_{\mathcal{S}}(u)\}$.

This is a special case of the more general {\MLA} problem.
Given a clustering $\mathcal{S}$, consider the vector of all agreement values $\text{agree}_{\mathcal{S}}(V)\in \mathcal{R}_+^V$, where $( \text{agree}_{\mathcal{S}}(V))_u=\text{agree}_{\mathcal{S}}(u)$ $\forall u\in V$. 
The objective of {\MLA} is to find a partition $\mathcal{S}$ that maximizes $g( \text{agree}_{\mathcal{S}}(V))$ for a given function $g$.
For example, if $g$ is the $\min$ function {\MLA} reduces to {\MMA}, and if $g$ is the summation function {\MLA} reduces to the classic objective of maximizing total agreements.


{\MLA} is closely related to the computation of local optima for {\MC}, and the computation of pure Nash equilibria in cut and party affiliation games \cite{balcan2009improved,bhalgat2010approximating,christodoulou2006convergence,fabrikant2004complexity,schaffer1991simple} (a well studied special class of potential games \cite{monderer1996potential}).
In the setting of party affiliation games, each node of $G$ is a player that can choose one of two sides of a cut.
The player's payoff is the total weight of edges incident on it that are classified correctly.
It is well known that such games admit a pure Nash equilibria via the {\em best response dynamics} (also known as {\em Nash dynamics}), and that each such pure Nash equilibrium is a $\left(\nicefrac[]{1}{2}\right)$-approximation for {\MLA}.
Unfortunately, in general the computation of a pure Nash equilibria in cut and party affiliation games is PLS-complete \cite{johnson1988easy}, and thus it is widely believed no polynomial time algorithm exists for solving this task.
Nonetheless, one can apply the algorithm of Bhalgat {\em et al.} \cite{bhalgat2010approximating} for finding an approximate pure Nash equilibrium and obtain a $\nicefrac[]{1}{(4+\varepsilon)}$-approximation for {\MLA} (for any constant $\varepsilon > 0$).
This approximation is also the best known for the special case of {\MMA}.


\vspace{5pt}
\noindent {\bf{Our Results:}}
Focusing first on {\MMD} on general graphs we prove that both the natural LP and SDP relaxations admit a large integrality gap of $\nicefrac[]{n}{2}$. 
Nonetheless, we present an $O(\sqrt{n})$-approximation for {\MMD}, bypassing the above integrality gaps.


\begin{theorem}\label{thrm:IntegralityGapMMD}
The natural LP and SDP relaxations for {\MMD} have an integrality gap of $\nicefrac[]{n}{2}$.
\end{theorem}

\begin{theorem}\label{thrm:sqrtMMD}
{\MMD} admits an $O(\sqrt{n})$-approximation for general weighted graphs.
\end{theorem}

\noindent Since {\MMST}, along with {\MMMW} and {\MMM}, are a special case of {\MMD}, Theorem \ref{thrm:sqrtMMD} applies to them as well, thus providing the first known approximation for this family of cut problems.\footnote{Theorem \ref{thrm:IntegralityGapMMD} can be easily adapted to apply also for {\MMST}, {\MMMW}, and {\MMM}, resulting in a gap of $\nicefrac[]{(n-1)}{2}$.}

When considering the more general {\MLD} problem, we present a remarkably simple approach that achieves an improved approximation of $7$ for both complete graphs and complete bipartite graphs (where disagreements are measured w.r.t one side only). 
This improves upon and simplifies \cite{puleo2016correlation} who presented an approximation of $48$ for the former and $10$ for the latter.
\begin{theorem}\label{thrm:7ApproxCliqueMLD}
{\MLD} admits a $7$-approximation for complete graphs.\\
where $f$ is required to satisfy the following three conditions:
$(1)$ for any $\bx, \by\in \mathcal{R}_+^V$ if $\bx \leq \by$ then $f(\bx)\leq f(\by)$ (monotonicity), $(2)$ $f(\alpha \bx)\leq \alpha f(\bx)$ for any $\alpha \geq 0$ and $\bx\in \mathcal{R}_+^V$ (scaling), and $(3)$ $f$ is convex.

\end{theorem}
\begin{theorem}\label{thrm:7ApproxBipartiteMLD}
{\MLD} admits a $7$-approximation for complete bipartite graphs where disagreements are measured w.r.t. one side of the graph.
where $f$ is required to satisfy the following three conditions:
$(1)$ for any $\bx, \by\in \mathcal{R}_+^V$ if $\bx \leq \by$ then $f(\bx)\leq f(\by)$ (monotonicity), $(2)$ $f(\alpha \bx)\leq \alpha f(\bx)$ for any $\alpha \geq 0$ and $\bx\in \mathcal{R}_+^V$ (scaling), and $(3)$ $f$ is convex.
\end{theorem}

%


Focusing on local maximization of agreements, 
we present a $\nicefrac[]{1}{(2+\varepsilon)}$ approximation for {\MMA} without any assumption on the edge weights.
This improves upon the previous known $\nicefrac[]{1}{(4+\varepsilon)}$-approximation that follows from the computation of approximate pure Nash equilibria in party affiliation games \cite{bhalgat2010approximating}.
As before, we show that both the natural LP and SDP relaxations for {\MMA} have a large integrality gap of $\frac{n}{2(n-1)}$. 
\begin{theorem}\label{thrm:ApproxMMA}
For any $\varepsilon > 0$, {\MMA} admits a $\nicefrac[]{1}{(2+\varepsilon)}$-approximation for general weighted graphs, where the running time of the algorithm is $poly(n,\nicefrac[]{1}{\varepsilon})$. 
\end{theorem}
\begin{theorem}\label{thrm:IntegralityGapMMA}
The natural LP and SDP relaxations for {\MMA} have an integrality gap of $\frac{n}{2(n-1)}$.
\end{theorem}


\begin{table}[h]
\caption{Results for {\CC} with local guarantees.\label{tab:Results}}
\begin{center}
{\footnotesize
{
\begin{tabu}{ |[1pt]c|[1pt]c|[1pt]c|[1pt]c| }
\hline
\multirow{2}{*}{\bf Problem} & \multirow{2}{*}{\bf Input Graph} & \multicolumn{2}{ c| }{\bf Approximation} \\
\cline{3-4}
& & {\bf This Work} & {\bf Previous Work} \\
\hline
\multirow{2}{*}{\MLD} & complete & $7$ & $48$ \cite{puleo2016correlation} \\
\cline{2-4}
& complete bipartite (one sided) & $7$ & $10$ \cite{puleo2016correlation} \\
\hline
{\MMD} & general weighted & $O(\sqrt{n})$ & $-$ \\
\hline
{\hspace{-15pt}\MMST} & \multirow{3}{*}{general weighted} & \multirow{3}{*}{$O(\sqrt{n})$} & \multirow{3}{*}{$-$}\\
{\MMMW} & & &\\
{\hspace{-20pt}\MMM} & & &\\
\hline
{\MMA} & general weighted & $\nicefrac[]{1}{(2+\varepsilon)}$ & $\nicefrac[]{1}{(4+\varepsilon)}$ \cite{bhalgat2010approximating} \\
\hline
\end{tabu}
}
}

\end{center}
\end{table}
\noindent Our main algorithmic results are summarized in Table \ref{tab:Results}.

\vspace{5pt}
\noindent {\bf{Approach and Techniques:}} 
The non-linear nature of {\CC} with local guarantees makes problems in this family much harder to approximate than {\CC} with classic global objectives.

Firstly, LP and SDP relaxations are not always useful when considering local objectives.
For example, the natural LP relaxation for the global objective of minimizing total disagreements on general graphs has a bounded integrality gap of $O(\log{n})$ \cite{charikar2003clustering,demaine2006correlation,garg1993approximate}.
However, we prove that for its local objective counterpart, \ie, {\MMD}, both the natural LP and SDP relaxations have a huge integrality gap of $\nicefrac[]{n}{2}$ (Theorem \ref{thrm:IntegralityGapMMD}).
To overcome this our algorithm for {\MMD} on general weighted graphs uses a {\em combination} of the LP lower bound and a combinatorial bound.
Even though each of these bounds on its own is bad, we prove that their combination suffices to obtain an approximation of $O(\sqrt{n})$, thus bypassing the huge integrality gaps of $\nicefrac[]{n}{2}$.

Secondly, randomization is inherently difficult to use for local guarantees, while many of the algorithms for minimizing total disagreements, \eg, \cite{ailon2012improved,ailon2008aggregating,chawla2015near}, as well as maximizing total agreements, \eg, \cite{swamy2004correlation}, are all randomized in nature.
The reason is that a bound on the expected weight of misclassified edges incident on any node does not translate to a bound on the maximum of this quantity over all nodes (similarly the expected weight of correctly classified edges incident on any node does not translate to a bound on the minimum of this quantity over all nodes).
To overcome this difficulty, all the algorithms we present are deterministic, \eg, for {\MLD} we propose a new remarkably simple method of clustering that greedily chooses a center node $s^*$ and cuts a sphere of a fixed and predefined radius around $s^*$, and for {\MMA} we present a new {\em non-oblivious} local search algorithm that
runs on a graph with modified edge weights and
circumvents the need to compute approximate pure Nash equilibria in party affiliation games.

\vspace{5pt}
\noindent {\bf{Paper Organization:}}
Section \ref{sec:MinLocalDisagree} contains the improved approximations for {\MMD} on general weighted graphs and for {\MLD} on complete and complete bipartite graphs (Theorems \ref{thrm:sqrtMMD}, \ref{thrm:7ApproxCliqueMLD}, and \ref{thrm:7ApproxBipartiteMLD}), along with the integrality gaps of the natural LP and SDP relaxations (Theorem \ref{thrm:IntegralityGapMMD}).
Section \ref{sec:MaxMinAgree} contains the improved approximation for {\MMA} as well as the integrality gaps of the natural LP and SDP relaxations (Theorems \ref{thrm:ApproxMMA} and \ref{thrm:IntegralityGapMMA}).

\section{Preliminaries}
We state the conditions required from both $f$ and $g$ in the definitions of {\MLD} and {\MLA}, respectively.
$f$ is required to satisfy the following three conditions:
$(1)$ for any $\bx, \by\in \mathcal{R}_+^V$ if $\bx \leq \by$ then $f(\bx)\leq f(\by)$ (monotonicity), $(2)$ $f(\alpha \bx)\leq \alpha f(\bx)$ for any $\alpha \geq 0$ and $\bx\in \mathcal{R}_+^V$ (scaling), and $(3)$ $f$ is convex.
Whereas, $g$ is required to satisfy the following two conditions:
$(1)$ for any $\bx, \by\in \mathcal{R}_+^V$ if $\bx \leq \by$ then $g(\bx)\leq g(\by)$ (monotonicity), and $(2)$ $g(\alpha \bx)\geq \alpha g(\bx)$ for any $\alpha \geq 0$ and $\bx\in \mathcal{R}_+^V$ (reverse scaling).
Note that $g$ is not required to be concave.

\section{Local Minimization of Disagreements and Graph Cuts}\label{sec:MinLocalDisagree}

We consider the natural convex programming relaxation for {\MLD}. 
The relaxation imposes a metric $d$ on the vertices of the graph.
For each node $u\in V$ we have a variable $D(u)$ denoting the total fractional {\em disagreement} of edges incident on $u$. 
Additionally, we denote by $\bD\in \mathcal{R}_+^V$ the vector of all $D(u)$ variables.
Note that the relaxation is solvable in polynomial time since $f$ is convex.\footnote{The convexity of $f$ is used only to show that relaxation (\ref{Relaxation:Disagreements}) can be solved, and it is not required in the rounding process.}

\begin{align}
\min ~~~ & f\left( \bD\right) & \label{Relaxation:Disagreements}\\
& \sum _{v:(u,v)\in E^+}c_{u,v} d\left( u,v\right) + \sum _{v:(u,v)\in E^-} c_{u,v} \left( 1-d\left( u,v\right) \right) = D(u) & \forall u\in V \nonumber \\
& d(u,v) + d(v,w) \geq d(u,w) & \forall u,v,w\in V \nonumber \\
& D(u)\geq 0, ~0\leq d(u,v) \leq 1 & \forall u,v\in V \nonumber
\end{align}

\noindent For the special case of {\MMD}, \ie, $f$ is the $\max$ function, (\ref{Relaxation:Disagreements}) can be written as an LP.
The proof of Theorem \ref{thrm:IntegralityGapMMD}, which states that even for the special case of {\MMD} the above natural LP and in addition the natural SDP both have a large integrality gap of $\nicefrac[]{n}{2}$, appears in Appendix \ref{app:GapMMD}.
We note that Theorem \ref{thrm:IntegralityGapMMD} also applies to {\MMST}, a further special case of {\MMD}.

\subsection{Min Max Disagreements on General Weighted Graphs}


Our algorithm for {\MMD} on general weighted graphs cannot rely solely on the the lower bound of the LP relaxation, since it admits an integrality gap of $\nicefrac[]{n}{2}$ (Theorem \ref{thrm:IntegralityGapMMD}).
Thus, a different lower bound must be used.
Let {$\Cmax$} be the maximum weight of an edge that is misclassified in some optimal solution $\mathcal{S}^*$.
Clearly, {$\Cmax$} also serves as a lower bound on the value of an optimal solution.
Hence, we can mix these two lower bounds and choose $\max \left\{ \max_{u\in V}\left\{ D(u)\right\},\Cmax\right\}$ to be the lower bound we use.
Note that we can assume w.l.o.g. that {$\Cmax$} is known to the algorithm, as one can simply execute the algorithm for every possible value of {$\Cmax$} and return the best solution.

Our algorithm consists of two main phases.
In the first we compute the LP metric $d$ but require additional constraints that ensure no {\em heavy} edge, \ie, an edge $e$ having $ c_e > \Cmax$, is (fractionally) misclassified by $d$.
In the second phase, we perform a careful {\em layered clustering} of an auxiliary graph consisting of all $+$ edges whose length in the metric $d$ is short.
At the heart of the analysis lies a distinction between $+$ edges whose length in the metric $d$ is short and all other edges.
The contribution of the former is bounded using the combinatorial lower bound, \ie, $\Cmax$, whereas the contribution of the latter is bounded using the LP.
Our algorithm also ensures that in the final clustering no heavy edge is misclassified.
Let us now elaborate on the two phases, before providing an exact description of the algorithm (Algorithm \ref{alg:MMD_General}).

\noindent {\bf{Phase $1$ (constrained metric computation):}}
Denote by, $$E^+_{\text{heavy}}\triangleq \{ e\in E^+:c_e>\Cmax\} ~~~\text{ and }~~~ E^-_{\text{heavy}}\triangleq \{ e\in E^-:c_e>\Cmax\}$$ the collection of all heavy $+$ and $-$ edges, respectively.
We solve the LP relaxation (\ref{Relaxation:Disagreements}) (recall that $f$ is the $\max$ function) while adding the following additional constraints that ensure $d$ does not (fractionally) misclassify heavy edges:
\begin{align}
d(u,v) = 0 & ~~~~~~~~~~~\forall e=(u,v)\in E^+_{\text{heavy}} \label{extraConst1} \\
d(u,v) = 1 & ~~~~~~~~~~~\forall e=(u,v)\in E^-_{\text{heavy}} \label{extraConst2}
\end{align}
If no feasible solution exists then our current guess for $\Cmax$ is incorrect.

\noindent {\bf{Phase $2$ (layered clustering):}}
Denote the collections of $+$ and $-$ edges which are {\em almost} classified correctly by $d$ as $E^+_{\text{bad}}\triangleq \left\{ e=(u,v)\in E^+:d(u,v)<\nicefrac[]{1}{\sqrt{n}}\right\}$ and $E^-_{\text{bad}}\triangleq \left\{ e=(u,v)\in E^-:d(u,v)>1-\nicefrac[]{1}{\sqrt{n}}\right\}$, respectively.
Intuitively, any edge $ e\notin E^+_{\text{bad}}\cup E^-_{\text{bad}}$ can use its length $d$ to pay for its contribution to the cost, regardless of what the output is. 
This is not the case with edges in $E^+_{\text{bad}} $ and $E^-_{\text{bad}} $, therefore all such edges are considered {\em bad}.
Additionally, denote by $E^+_0\triangleq \left\{ e=(u,v)\in E^+:d(u,v)=0\right\}$ the collection of $+$ edges for which $d$ assigns a length of $0$.\footnote{Note that $ E^+_{\text{heavy}}\subseteq E^+_0\subseteq E^+_{\text{bad}}$ and $E^-_{\text{heavy}}\subseteq  E^-_{\text{bad}}$.}

We design the algorithm so it ensures that no mistakes are made for edges in $ E^+_0$ and $E^-_{\text{bad}} $.
However, the algorithm might make mistakes for edges in $ E^+_{\text{bad}}$, thus a careful analysis is required.
To this end we consider the auxiliary graph consisting of all edges in $E^+_{\text{bad}}$, \ie, $G^+_{\text{bad}}\triangleq \left( V,E^+_{\text{bad}}\right)$, and equip it with the distance function $\text{dist}_{\ell}$ defined as the shortest path metric with respect to the length function $\ell:E^+_{\text{bad}}\rightarrow\left\{ 0,1\right\}$:
\begin{align}
\ell (e) \triangleq \begin{cases}0 & e\in E^+_0\\ 1 & e\in E^+_{\text{bad}}\setminus E^+_0\end{cases}\nonumber
\end{align}
Assume $E^-_{\text{bad}} $ contains $k$ edges and denote the endpoints of the $i$\textsuperscript{th} edge by $s_i$ and $t_i$.
The algorithm partitions every connected component $X$ of $G^+_{\text{bad}}$ into clusters as follows: as long as $X$ contains $s_i$ and $t_i$ for some $i$, we examine the layers $\text{dist}_{\ell}(s_i,\cdot)$ defines and perform a carefully chosen level cut.
This {\em layered clustering} suffices as we can prove that our choice of a level cut ensures $(1)$  no mistakes are made for edges in $ E^+_0$ and $E^-_{\text{bad}} $, and $(2)$ the {\em number} of misclassified edges from $E^+_{\text{bad}}\setminus E^+_0$ incident on any node is at most $O(\sqrt{n})$.
This ends the description of the second phase.

\begin{algorithm}
\caption{Layered Clustering $(G=(V,E),\Cmax)$}\label{alg:MMD_General}
\begin{algorithmic}[1]
\STATE $\mathcal{C}\leftarrow \emptyset$.
\STATE let $d$ be a solution to LP (\ref{Relaxation:Disagreements}) with the additional constraints (\ref{extraConst1}) and (\ref{extraConst2})
\FOR {every connected component $X$ in $G^+_{\text{bad}}$}
\WHILE {$X$ contains $\left\{ s_i,t_i\right\}$ for some $i$}
\STATE $ r_i\leftarrow \text{dist}_{\ell}(s_i,t_i)$ and $L^i_j\leftarrow \left\{ u:\text{dist}_{\ell}(s_i,u)=j\right\}$ for every $j=0,1,\ldots, r_i$.
\STATE choose $j^*\leq\nicefrac[]{(\sqrt{n}-1)}{2}$ s.t. $| L^i_{j^*}|,| L^i_{j^*+1}|,| L^i_{j^*+2}|\leq 16\sqrt{n}$.
\STATE $S\leftarrow \cup _{j=0}^{j^*}L^i_j$.
\STATE $X\leftarrow X\setminus S$ and $\mathcal{C}\leftarrow \mathcal{C}\cup \{ S\} $.
\ENDWHILE
\STATE $\mathcal{C}\leftarrow \mathcal{C}\cup \left\{ X\right\}$.
\ENDFOR
\STATE Output $\mathcal{C}$.
\end{algorithmic}
\end{algorithm}

Refer to Algorithm \ref{alg:MMD_General} for a precise description of the algorithm.
The following Lemma states that the distance between any $\{ s_i,t_i\}$ pair with respect to the metric $\text{dist}_{\ell}$ is large, its proof appears in Appendix \ref{app:LongPath}.
\begin{lemma}\label{lem:LongPath}
For every $i=1,\ldots,k$, $\text{dist}_{\ell}(s_i,t_i)>\sqrt{n}-1$.
\end{lemma}
%

\noindent The following Lemma simply states that only a few layers could be too large, its proof appears in Appendix \ref{app:BadLayers}.
It implies Corollary \ref{cor:ChoosingLayer}, whose proof appears in Appendix \ref{app:ChoosingLayer}.
\begin{lemma}\label{lem:BadLayers}
For every $i=1,\ldots,k$, the number of layers $L^i_j$ for which $|L^i_j|> 16\sqrt{n}$ is at most $\nicefrac[]{\sqrt{n}}{16}$.
\end{lemma}

\begin{corollary}\label{cor:ChoosingLayer}
Algorithm \ref{alg:MMD_General} can always find $j^*$ as required.
\end{corollary}

\noindent Lemma \ref{lem:CorrectEdges} proves that no mistakes are made for edges in $E^+_0 $ and $ E^-_{\text{bad}}$, whereas Lemma \ref{lem:Short+Edges} bounds the {\em number} of misclassified edges from $E^+_{\text{bad}}\setminus E^+_0 $ incident on any node.
Their proofs appear in Appendices \ref{app:CorrectEdges} and \ref{app:Short+Edges}.
\begin{lemma}\label{lem:CorrectEdges}
Algorithm \ref{alg:MMD_General} never misclassifies edges in $E^+_0 $ and $ E^-_{\text{bad}}$.
\end{lemma}
%

\begin{lemma}\label{lem:Short+Edges}
Let $u\in V$ and $S$ be the cluster in $\mathcal{C}$ Algorithm \ref{alg:MMD_General} assigned $u$ to.
Then, $\left|\left\{ e\in E^+_{\text{bad}}\setminus E^+_0:e=(u,v),v\notin S\right\}\right|\leq 48\sqrt{n}$.
\end{lemma}
%
%
\noindent We are now ready to prove the main result, Theorem \ref{thrm:sqrtMMD}.
\begin{proof}[of Theorem \ref{thrm:sqrtMMD}]
We prove that Algorithm \ref{alg:MMD_General} achieves an approximation of $49\sqrt{n}$.
The proof considers edges according to their type: $(1)$ $E^+_0$ and $E^-_{\text{bad}}$ edges, $(2)$ $E^+_{\text{bad}}\setminus E^+_0$ edges, and $(3)$ all other edges.
It is worth noting that the contribution of edges of type $(2)$ is bounded using the combinatorial lower bound, \ie, $\Cmax$, whereas the contribution of edges of type $(3)$ is bounded using the LP, \ie, $D(u)$ for every node $u\in V$ (as defined by the relaxation (\ref{Relaxation:Disagreements})).

First, consider edges of type $(1)$.
Lemma \ref{lem:CorrectEdges} implies Algorithm \ref{alg:MMD_General} does not make any mistakes with respect to these edges, thus their contribution to the value of the output $\mathcal{C}$ is always $0$.
Second, consider edges of type $(2)$.
Lemma \ref{lem:Short+Edges} implies that every node $u$ has at most $48\sqrt{n}$ edges of type $(2)$ incident on it that are classified incorrectly.
Additionally, the weight of every edge of type $(2)$ is at most $\Cmax$ since $ E^+_{\text{heavy}}\subseteq E^+_0$ and edges of type $(2)$ do not contain any edge of $E^+_0$.
Thus, we can conclude that for every node $u$ the total weight of edges of type $(2)$ that touch $u$ and are misclassified is at most $48\sqrt{n}\cdot \Cmax$.

Finally, consider edges of type $(3)$.
Fix an arbitrary node $u$ and let $D(u)$ be the fractional disagreement value the LP assigned to $u$ (see (\ref{Relaxation:Disagreements})).
Edge $e$ of type $(3)$ is either an edge $e\in E^+$ whose $d$ length is at least $\nicefrac[]{1}{\sqrt{n}}$, or an edge $e\in E^-$ whose $d$ length is at most $1-\nicefrac[]{1}{\sqrt{n}}$.
Hence, in any case the fractional contribution of such an edge $e$ to $D(u)$ is at least $\nicefrac[]{c_e}{\sqrt{n}}$.
Therefore, regardless of what the output is, the total weight of misclassified edges of type $(3)$ incident on $u$ is at most $\sqrt{n}\cdot D(u)$.

Summing over all types of edges, we can conclude that the total weight of misclassified edges incident on $u$ in $\mathcal{C}$ (the output of Algorithm \ref{alg:MMD_General}) is at most $48\sqrt{n}\Cmax + \sqrt{n}\cdot D(u)$.
Since both $\Cmax$ and $D(u)$ are lower bounds on the value of an optimal solution, the proof is concluded. \hfill $\square$ 
\end{proof}


\subsection{Min Local Disagreements on Complete Graphs}
We consider a simple deterministic greedy clustering algorithm for complete graphs that iteratively partitions the graph.
In every step it does the following: $(1)$ greedily chooses a center node $s^*$ that has many nodes {\em close} to it, and $(2)$ removes from the graph a sphere around $s^*$ which constitutes a new cluster. 
The greedy choice of $s^*$ is similar to that of \cite{puleo2016correlation}.
However, our algorithm departs from the approach of \cite{puleo2016correlation}, as it {\em always} cuts a large sphere around $s^*$.
The algorithm of \cite{puleo2016correlation}, on the other hand, outputs either a singleton cluster containing $s^*$ or some other large sphere around $s^*$ (the average distance within the large sphere determines which of the two options is chosen), thus mimicking the approach of \cite{charikar2003clustering}.
Surprisingly, restricting the algorithm's choice enables us not only to obtain a simpler algorithm, but also to improve upon the approximation guarantee from $48$ to $7$.

Algorithm \ref{alg:7ApproxClique} receives as input the metric $d$ as computed by the relaxation (\ref{Relaxation:Disagreements}), whereas the variables $D(u)$ are required only for the analysis.
Additionally, we denote by $\text{Ball}_S(u,r)\triangleq \left\{ v\in S:d(u,v)<r\right\}$ the sphere of radius $r$ around $u$ in subgraph $S$.

\begin{algorithm}
\caption{Greedy Clustering $( \{ d(u,v)\} _{u,v\in V})$}\label{alg:7ApproxClique}
\begin{algorithmic}[1]
\STATE $S\leftarrow V$ and $\mathcal{C}\leftarrow \emptyset$.
\WHILE {$S\neq \emptyset$}
\STATE $s^*\leftarrow \text{argmax}\left\{ \left| \text{Ball}_S(s,\nicefrac[]{1}{7})\right|:s\in S\right\}$.
\STATE $\mathcal{C} ~\leftarrow \mathcal{C} \cup \left\{ \text{Ball}_S(s^*,\nicefrac[]{3}{7})\right\}$.
\STATE $S~\leftarrow S\setminus \text{Ball}_S(s^*,\nicefrac[]{3}{7})$.
\ENDWHILE
\STATE Output $\mathcal{C}$.
\end{algorithmic}
\end{algorithm}
\noindent The following lemma summarizes the guarantee achieved by Algorithm \ref{alg:7ApproxClique} (its proof appears in Appendix \ref{app:7ApproxClique}, which also contains an overview of our charging scheme).
\begin{lemma}\label{lem:7ApproxClique}
Assuming the input is a complete graph, Algorithm \ref{alg:7ApproxClique} guarantees that $\text{disagree}_{\mathcal{C}}(u) \leq 7D(u)$ for every $u\in V$.
\end{lemma}

\begin{proof}[of Theorem \ref{thrm:7ApproxCliqueMLD}]
Apply Algorithm \ref{alg:7ApproxClique} to the solution of the relaxation (\ref{Relaxation:Disagreements}).
Lemma \ref{lem:7ApproxClique} guarantees that for every node $u\in V$ we have that $ \text{disagree}_{\mathcal{C}}(u) \leq 7D(u)$, \ie, $ \text{disagree}_{\mathcal{C}}(V)\leq 7\bD$.
The value of the output of the algorithm is $f\left( \text{disagree}_{\mathcal{C}}(V)\right)$ and one can bound it as follows:
$$ f\left( \text{disagree}_{\mathcal{C}}(V)\right)\stackrel{(1)}{\leq} f\left( 7\bD\right) \stackrel{(2)}{\leq} 7f\left( \bD\right)~.$$
Inequality $(1)$ follows from the monotonicity of $f$, whereas inequality $(2)$ follows from the scaling property of $f$.
This concludes the proof since $f\left( \bD\right) $ is a lower bound on the value of any optimal solution. \hfill $\square$
\end{proof}

\subsection{Min Local Disagreements on Complete Bipartite Graphs}
Our algorithm for {\MLD} on complete bipartite graphs (with one sided disagreements) is a natural extension of Algorithm \ref{alg:7ApproxClique}.
Similarly to the complete graph case, we are able to present a remarkably simple algorithm achieving an improved approximation of $7$.
The description of the algorithm and the proof of Theorem \ref{thrm:7ApproxBipartiteMLD} appear in Appendix \ref{app:7Bipartite}.

\section{Local Maximization of Agreements}\label{sec:MaxMinAgree}

As previously mentioned, {\MLA} is closely related to the computation of local optima for {\MC} and pure Nash equilibria in cut and party affiliation games, both of which are PLS-complete problems.
We focus on the special case of {\MMA}.

The natural local search algorithm for {\MMA} can be defined similarly to that of {\MC}:
it maintains a single cut $S\subseteq V$; a node $u$ moves to the other side of the cut if the move increases the total weight of correctly classified edges incident on $u$.
This algorithm terminates in a local optimum that is a $\left(\nicefrac[]{1}{2}\right)$-approximation for {\MMA}.
Unfortunately, it is known that such a local search algorithm 
can take exponential time,
even for {\MC}.

When considering {\MC}, this can be remedied by altering the local search step as follows: a node $u$ moves to the other side of the cut $S$ if 
the move increases the total weight of edges crossing $S$ 
by a multiplicative factor of at least $(1+\varepsilon)$ (for some $\varepsilon > 0$).
This approach {\em fails} for the computation of (approximate) pure Nash equilibria in party affiliation games, as well as for {\MMA}.
The reason is that both of these problems have {\em local} requirements from nodes, as opposed to the {\em global} objective of {\MC}.
Thus, not surprisingly, the current best known $\nicefrac[]{1}{(4+\varepsilon)}$-approximation for {\MMA} follows from \cite{bhalgat2010approximating} who present the state of the art algorithm for finding approximate pure Nash equilibria in party affiliation games.

We propose a direct approach for approximating {\MMA} that circumvents the need to compute approximate pure Nash equilibria in party affiliation games.
We improve upon the $\nicefrac[]{1}{(4+\varepsilon)}$-approximation by considering a {\em non-oblivious} local search that is executed with altered edge weights.
We are able to change the edges' weights in such a way that: $(1)$ any local optimum is a $\nicefrac[]{1}{(2+\varepsilon)}$-approximation, and $(2)$ the local search performs at most $O(\nicefrac[]{n}{\varepsilon})$ iterations.
The proof of Theorem \ref{thrm:ApproxMMA} appears in Appendix \ref{app:ApproxMMA}, along with some intuition for our non-oblivious local search algorithm.
Additionally, we prove that the natural LP and SDP relaxations for {\MMA} on general graphs admit an integrality gap of $\frac{n}{2(n-1)}$ (Theorem \ref{thrm:IntegralityGapMMA}).
This appears in Appendix \ref{app:IntegralityGapMMA}.

\bibliography{refs}

\begin{thebibliography}{10}
\providecommand{\url}[1]{\texttt{#1}}
\providecommand{\urlprefix}{URL }

\bibitem{ailon2012improved}
Ailon, N., Avigdor-Elgrabli, N., Liberty, E., van Zuylen, A.: Improved
  approximation algorithms for bipartite correlation clustering. SIAM Journal
  on Computing  41(5),  1110--1121 (2012)

\bibitem{ailon2008aggregating}
Ailon, N., Charikar, M., Newman, A.: Aggregating inconsistent information:
  ranking and clustering. Journal of the ACM (JACM)  55(5), ~23 (2008)

\bibitem{amit2004bicluster}
Amit, N.: The bicluster graph editing problem. Ph.D. thesis, Tel Aviv
  University (2004)

\bibitem{balcan2009improved}
Balcan, M.F., Blum, A., Mansour, Y.: Improved equilibria via public service
  advertising. In: SODA' 09. pp. 728--737 (2009)

\bibitem{bansal2004correlation}
Bansal, N., Blum, A., Chawla, S.: Correlation clustering. Machine Learning
  56(1-3),  89--113 (2004)

\bibitem{bansal2014min}
Bansal, N., Feige, U., Krauthgamer, R., Makarychev, K., Nagarajan, V., Naor,
  J., Schwartz, R.: Min-max graph partitioning and small set expansion. SIAM
  Journal on Computing  43(2),  872--904 (2014)

\bibitem{ben1999clustering}
Ben-Dor, A., Shamir, R., Yakhini, Z.: Clustering gene expression patterns.
  Journal of computational biology  6(3-4),  281--297 (1999)

\bibitem{bhalgat2010approximating}
Bhalgat, A., Chakraborty, T., Khanna, S.: Approximating pure nash equilibrium
  in cut, party affiliation, and satisfiability games. In: EC' 10. pp. 73--82
  (2010)

\bibitem{charikar2003clustering}
Charikar, M., Guruswami, V., Wirth, A.: Clustering with qualitative
  information. In: FOCS' 03. pp. 524--533 (2003)

\bibitem{chawla2015near}
Chawla, S., Makarychev, K., Schramm, T., Yaroslavtsev, G.: Near optimal {LP}
  rounding algorithm for correlationclustering on complete and complete
  k-partite graphs. In: STOC' 15. pp. 219--228 (2015)

\bibitem{cheng2000biclustering}
Cheng, Y., Church, G.M.: Biclustering of expression data. In: Ismb. vol.~8, pp.
  93--103 (2000)

\bibitem{christodoulou2006convergence}
Christodoulou, G., Mirrokni, V.S., Sidiropoulos, A.: Convergence and
  approximation in potential games. In: STACS' 06. pp. 349--360 (2006)

\bibitem{demaine2006correlation}
Demaine, E.D., Emanuel, D., Fiat, A., Immorlica, N.: Correlation clustering in
  general weighted graphs. Theoretical Computer Science  361(2),  172--187
  (2006)

\bibitem{fabrikant2004complexity}
Fabrikant, A., Papadimitriou, C., Talwar, K.: The complexity of pure nash
  equilibria. In: Proceedings of the thirty-sixth annual ACM symposium on
  Theory of computing. pp. 604--612. ACM (2004)

\bibitem{filkov2004integrating}
Filkov, V., Skiena, S.: Integrating microarray data by consensus clustering.
  International Journal on Artificial Intelligence Tools  13(04),  863--880
  (2004)

\bibitem{garg1993approximate}
Garg, N., Vazirani, V.V., Yannakakis, M.: Approximate max-flow min-(multi) cut
  theorems and their applications. In: STOC' 93. pp. 698--707 (1993)

\bibitem{johnson1988easy}
Johnson, D.S., Papadimitriou, C.H., Yannakakis, M.: How easy is local search?
  Journal of computer and system sciences  37(1),  79--100 (1988)

\bibitem{kriegel2009clustering}
Kriegel, H.P., Kr{\"o}ger, P., Zimek, A.: Clustering high-dimensional data: A
  survey on subspace clustering, pattern-based clustering, and correlation
  clustering. ACM Transactions on Knowledge Discovery from Data (TKDD)  3(1),
  ~1 (2009)

\bibitem{monderer1996potential}
Monderer, D., Shapley, L.S.: Potential games. Games and economic behavior
  14(1),  124--143 (1996)

\bibitem{puleo2016correlation}
Puleo, G., Milenkovic, O.: Correlation clustering and biclustering with locally
  bounded errors. In: Proceedings of The 33rd International Conference on
  Machine Learning. pp. 869--877 (2016)

\bibitem{schaffer1991simple}
Sch{\"a}ffer, A.A., Yannakakis, M.: Simple local search problems that are hard
  to solve. SIAM journal on Computing  20(1),  56--87 (1991)

\bibitem{svitkina2004min}
Svitkina, Z., Tardos, {\'E}.: Min-max multiway cut. In: Approximation,
  Randomization, and Combinatorial Optimization. Algorithms and Techniques, pp.
  207--218. Springer (2004)

\bibitem{swamy2004correlation}
Swamy, C.: Correlation clustering: maximizing agreements via semidefinite
  programming. In: SODA' 04. pp. 526--527 (2004)

\bibitem{symeonidis2006nearest}
Symeonidis, P., Nanopoulos, A., Papadopoulos, A., Manolopoulos, Y.:
  Nearest-biclusters collaborative filtering with constant values. In:
  International Workshop on Knowledge Discovery on the Web. pp. 36--55.
  Springer (2006)

\bibitem{van2007correlation}
Van~Gael, J., Zhu, X.: Correlation clustering for crosslingual link detection.
  In: IJCAI. pp. 1744--1749 (2007)

\bibitem{Wirth10}
Wirth, A.: Correlation clustering. C. Sammut and G. Webb, editors, Encyclopedia
  of Machine Learning pp. 227–--231 (2010)

\end{thebibliography}

\appendix
\newpage

\section{Proof of Theorem \ref{thrm:IntegralityGapMMD}}\label{app:GapMMD}
\begin{proof}[of Theorem \ref{thrm:IntegralityGapMMD}]
Let $G$ be the unweighted cycle on $n$ vertices, where all edges are labeled $+$ and one edge is labeled $-$.
Specifically, denote the vertices of $G$ by $\left\{ v_1,v_2,\ldots,v_n\right\}$ where there is an edge $(v_i,v_{i+1})\in E^+$ for every $i=1,\ldots,n-1$ and additionally the edge $(v_n,v_1)\in E^-$.

First, we prove that the value of any integral solution is at least $1$.
A clustering that includes $V$ as a single cluster has value of $1$, as both $v_1$ and $v_n$ have exactly one misclassified edge touching them.
Moreover, one can easily verify that any clustering into two or more clusters has a value of at least $1$.
Thus, any integral solution for the above instance has value of at least $1$.

For simplicity of presentation let us re-state here the LP relaxation (\ref{Relaxation:Disagreements}) of {\MMD}:
\begin{align*}
\min ~~~ & \max _{u\in V}\left\{ D(u)\right\} & \\
& \sum _{v:(u,v)\in E^+}c_{u,v}d\left( u,v\right) + \sum _{v:(u,v)\in E^-} c_{u,v}\left( 1-d\left( u,v\right) \right) = D(u) & \forall u\in V \\
& d(u,v) + d(v,w) \geq d(u,w) & \forall u,v,w\in V \\
& D(u)\geq 0, ~0\leq d(u,v) \leq 1 & \forall u,v\in V
\end{align*}
Let us construct a fractional solution.
Assign a length of $\nicefrac[]{1}{n}$ for every $+$ edge and a length of $1-\nicefrac[]{1}{n}$ for the single $-$ edge, and let $d$ be the shortest path metric in $G$ induced by these lengths.
Obviously, the triangle inequality is satisfied and one can verify that $d(u,v)\leq 1$ for all $u,v\in V$.
Consider a vertex $v_i$ that does not touch the $-$ edge, \ie, $i=2,\ldots,n-1$.
Such a $v_i$ has two + edges touching it both having a length of $\nicefrac[]{1}{n}$, hence $D(v_i)=\nicefrac[]{2}{n}$.
Focusing on $v_1$ and $v_n$, each has one $+$ edge whose length is $\nicefrac[]{1}{n}$ and one $-$ edge whose length is $1-\nicefrac[]{1}{n}$ touching it.
Hence, $D(v_1)=D(v_n)=\nicefrac[]{2}{n}$.
Therefore, the above instance has an integrality gap of $\nicefrac[]{n}{2}$.

Now, consider the natural semi-definite relaxation for {\MMD}, where each vertex $u$ corresponds to a unit vector $\by _u$.
Intuitively, if $S_1,\ldots, S_{\ell}$ is an integral clustering, then all vertices in cluster $S_j$ are assigned to the standard $j$\textsuperscript{th} unit vector, \ie, $\be _j$.
Hence, the natural semi-definite relaxation requires that all vectors lie in the same orthant, \ie, for every $u$ and $v$: $\by _u \cdot \by _v\geq 0$, and that $\{ \by _u\} _{u\in V}$ satisfy the $\ell _2^2$ triangle inequality.
Therefore, the natural semi-definite relaxation is:
\begin{align*}
\min ~~~ & \max _{u\in V}\left\{ D(u)\right\} & \\
& \sum _{v:(u,v)\in E^+}c_{u,v}\left( 1-\by _u \cdot \by _v\right) + \sum _{v:(u,v)\in E^-} c_{u,v}\left(\by _u \cdot \by _v\right) = D(u) & \forall u\in V \\
& || \by _u - \by _v||_2^2 + || \by _v - \by _w||_2^2 \geq || \by _u - \by _w||_2^2 & \forall u,v,w\in V \\
& \by _u \cdot \by _u = 1 & \forall u\in V \\
& \by _u \cdot \by _v \geq 0 & \forall u,v\in V
\end{align*}

In order to construct a fractional solution, it will be helpful to consider $Y\in \mathcal{R}^{V\times V}$ the positive semi-definite matrix of all inner products of $\left\{ \by _{v_i}\right\} _{i=1}^n$, \ie, $Y_{v_i,v_j}=\by _{v_i}\cdot \by _{v_j}$.
Intuitively, we consider a collection of integral solutions where for each one we construct the corresponding $Y$ matrix.
At the end, our fractional solution will be the average of all these $Y$ matrices.

Consider the following $n-1$ integral solutions, each having only two clusters, where the first cluster consists of $\left\{ v_1,\ldots,v_i\right\}$ and the second consists of $\left\{ v_{i+1},\ldots,v_n\right\}$ (here $i=1,\ldots,n-1$).
Fixing $i$ and using the above translation of an integral solution to a feasible solution for the semi-definite relaxation, we assign each $v_j$, where $j=1,\ldots,i$ to $\be _1$ and each $v_j$, where $j=i+1,\ldots,n$, to $\be _2$.
Let $Y^i$ be the resulting (positive semi-definite) inner product matrix.
Additionally, consider one additional integral solution that consists of a single cluster containing all of $V$.
In this case, the above translation yields that all $v_i$ vectors are assigned to $\be _1$.
Denote by $Y^n$ the resulting (positive semi-definite) inner product matrix.
Clearly, each of the $Y^1,\ldots,Y^n$ defines a feasible solution for the above natural semi-definite relaxation.

Our fractional solution is given by the average of all the above inner product matrices: $\overline{Y}\triangleq \frac{1}{n}\sum _{i=1}^n Y^i$.
Obviously, $\overline{Y}$ defines a feasible solution for the above natural semi-definite relaxation.
Note that  $\by _{v_1} \cdot \by _{v_n} = \frac{n-1}{n}\cdot 0 + \frac{1}{n}\cdot 1 = \frac{1}{n}$ and that $\by _{v_i} \cdot \by _{v_{i+1}} =\frac{n-1}{n}\cdot 1 + \frac{1}{n}\cdot 0 = \frac{n-1}{n}$, for every $i=1,\ldots,n-1$.
Therefore, we can conclude that:
\begin{align*}
D(v_i) & = 2\left( 1-\frac{n-1}{n}\right)=\frac{2}{n} &\forall i=2,\ldots,n-1 \\
D(v_1) & =D(v_n) =\left( 1-\frac{n-1}{n}\right) + \frac{1}{n}=\frac{2}{n} &
\end{align*}
This demonstrates that the above instance also has an integrality gap of $\nicefrac[]{n}{2}$ for the natural semi-definite relaxation. \hfill $\square$
\end{proof}

\section{Proof of Lemma \ref{lem:LongPath}}\label{app:LongPath}
\begin{proof}[of Lemma \ref{lem:LongPath}]
If $s_i$ and $t_i$ are not in the same connected component of $G^+_{\text{bad}}$ then $\text{dist}_{\ell}(s_i,t_i)=\infty$.
Otherwise, let $P$ be a path connecting $s_i$ and $t_i$ in $G^+_{\text{bad}}$.
Note that $\sum _{e=(u,v)\in P}d(u,v)\geq d(s_i,t_i)>1-\nicefrac[]{1}{\sqrt{n}}$, where the first inequality follows from the triangle inequality for $d$ and the second inequality from the fact that $(s_i,t_i)\in E^-_{\text{bad}}$, \ie, $d(s_i,t_i)>1-\nicefrac[]{1}{\sqrt{n}}$.

Let us now lower bound the number of edges in $P$ that belong to $ E^+_{\text{bad}}\setminus E^+_0$, \ie, edges $e$ for which $\ell(e)=1$.
Examine $\sum _{e=(u,v)\in P}d(u,v)$ and note that every edge $e=(u,v)\in E^+_0$ has a $d$ length of $0$.
Hence, we can remove those edges from the sum and conclude that: $ \sum _{e=(u,v)\in P\setminus E^+_0}d(u,v)>1-\nicefrac[]{1}{\sqrt{n}}$.
Recall that $G^+_{\text{bad}}$ contains only edges from $E^+_{\text{bad}}$, thus every $e=(u,v)\in P\setminus E^+_0$ satisfies: $d(u,v)<\nicefrac[]{1}{\sqrt{n}}$.
This implies that $P\setminus E^+_0$ must contain more than $ \sqrt{n}-1$ edges, \ie, $\text{dist}_{\ell}(s_i,t_i)>\sqrt{n}-1$, concluding the proof. \hfill $\square$
\end{proof}

\section{Proof of Lemma \ref{lem:BadLayers}}\label{app:BadLayers}
\begin{proof}[of Lemma \ref{lem:BadLayers}]
Let $x$ be the number of layers $L^i_j$ for which $|L^i_j|> 16\sqrt{n}$.
Since the total number of vertices in all layers, for a fixed $i$, cannot exceed $n$, we can conclude that $x\leq \nicefrac[]{n}{(16\sqrt{n})}=\nicefrac[]{\sqrt{n}}{16}$. \hfill $\square$
\end{proof}

\section{Proof of Corollary \ref{cor:ChoosingLayer}}\label{app:ChoosingLayer}
\begin{proof}[of Corollary \ref{cor:ChoosingLayer}]
We will prove that for any connected component $X$ of $G^+_{\text{bad}}$, such that both $\left\{ s_i,t_i\right\}$ belong to $X$, there are $3$ consecutive layers as required by Algorithm \ref{alg:MMD_General}.\footnote{This implies that Algorithm \ref{alg:MMD_General} can always find $j^*$ as required since any connected component $X$ can only shrink as the algorithm progresses.}
Lemma \ref{lem:BadLayers} implies that there are at most $\nicefrac[]{\sqrt{n}}{16}$ layers whose size is more than $ 16\sqrt{n}$.
Therefore, the number of layers among  $L^i_0,\ldots,L^i_{\nicefrac[]{(\sqrt{n}-1)}{2}}$ whose size is at most $16\sqrt{n}$ is at least:
$ \nicefrac[]{(\sqrt{n}-1)}{2} - \nicefrac[]{\sqrt{n}}{16}$.
The latter is at least $\nicefrac[]{3}{4}\cdot\nicefrac[]{(\sqrt{n}-1)}{2} $ (as long as $n\geq 4$).
Thus, there must be at least $3$ consecutive layers among $L^i_0,\ldots,L^i_{\nicefrac[]{(\sqrt{n}-1)}{2}}$, each having a size of at most $16\sqrt{n}$. \hfill $\square$
\end{proof}

\section{Proof of Lemma \ref{lem:CorrectEdges}}\label{app:CorrectEdges}
\begin{proof}[of Lemma \ref{lem:CorrectEdges}]
Let us start by focusing on edges in $E^+_0$.
Since $E^+_0\subseteq E^+_{\text{bad}}$, all edges of $ E^+_0$ are present in $G^+_{\text{bad}}$ by definition, and in particular both endpoints of every $e\in E^+_0$ are contained in the same connected component $X$ of $G^+_{\text{bad}}$.
Furthermore, for any $i$ such that both $\left\{ s_i,t_i\right\}$ are contained in $X$, both endpoints $e$ must also be contained in the same layer $L^i_j$ (for some $j$).
This follows from the fact that $\ell(e)=0$ for all edges $ e\in E^+_0$ and the definition of all layers $L^i_0,\ldots,L^i_{r_i}$.
Thus, both endpoints of every edge $e\in E^+_0$ are always in the same cluster $S$ in the output $ \mathcal{C}$, \ie, such an edge $e$ is never misclassified.

Let us now focus on edges in $E^-_{\text{bad}} $, and recall that our notation implies that $E^-_{\text{bad}}=\left\{ (s_i,t_i)\right\}_{i=1}^k$.
We prove that $(s_i,t_i)$ cannot be contained in some cluster $S\in \mathcal{C}$.
Note that each cluster $S$ is in fact a sphere of radius at most $\nicefrac[]{(\sqrt{n}-1)}{2}$ with respect to the metric $\text{dist}_{\ell}$.
Lemma \ref{lem:LongPath} states that $\text{dist}_{\ell}(s_i,t_i)>\sqrt{n}-1$, hence the triangle inequality for $\text{dist}_{\ell}$ implies that both $s_i$ and $t_i$ cannot be contained in the same cluster $S$.
Therefore, each $(s_i,t_i)$ edge is never misclassified. \hfill $\square$
\end{proof}

\section{Proof of Lemma \ref{lem:Short+Edges}}\label{app:Short+Edges}
\begin{proof}[of Lemma \ref{lem:Short+Edges}]
Fix $u$, $S$ the cluster Algorithm \ref{alg:MMD_General} assigned $u$ to, and $X$ the connected component of $G^+_{\text{bad}}$ defining $S$.
Consider the first iteration an edge $e=(u,v)\in E^+_{\text{bad}}\setminus E^+_0$ touching $u$ is misclassified by the algorithm.
Let $i$ correspond to the pair $\left\{ s_i,t_i\right\}$ considered in the above iteration and $j^*$ the index by which Algorithm \ref{alg:MMD_General} defined the cluster in the same iteration.

If the above occurs in an iteration where $S$ itself is added to $\mathcal{C}$, then it must be the case that $u\in L^i_{j^*}$ and $v\in L^i_{j^*+1}$.
Additionally, no other edges in $ E^+_{\text{bad}}\setminus E^+_0$ touching $u$ can be misclassified in subsequent iteration.
Therefore, in this case the total number of edges in $E^+_{\text{bad}}\setminus E^+_0$ touching $u$ that are misclassified can be upper bounded by $|L^i_{j^*+1}|\leq 16\sqrt{n}$.

Otherwise, the first iteration an edge $e=(u,v)\in E^+_{\text{bad}}\setminus E^+_0$ touching $u$ is misclassified by the algorithm is not the iteration in which $S$ itself is added to $\mathcal{C}$.
Thus, since the algorithm cuts between layers $L^i_{j^*}$ and $L^i_{j^*+1}$, it must be the case that $u\in L^i_{j^*+1}$.
Since edges in $G^+_{\text{bad}}$ can connect only vertices in the same or adjacent layers, the total degree of $u$ in $G^+_{\text{bad}}$ is at most $|L^i_{j^*}|+|L^i_{j^*+1}|+|L^i_{j^*+2}|-1$.
From the choice of $j^*$ the latter can be upper bounded by $48\sqrt{n}$. \hfill $\square$
\end{proof}

\section{Proof of Lemma \ref{lem:7ApproxClique}}\label{app:7ApproxClique}
Let us start with some intuition as to why $s^*$ is chosen greedily.
One of the goals of the analysis is to bound the contribution of $+$ edges crossing the boundary of the sphere around $s^*$.
Since those edges might have an extremely small fractional contribution w.r.t. the metric, \ie, their $d$ length is very short, we must charge their cost to other edges.
The fact that there are many vertices close to $s^*$, along with the fact that the graph is complete, implies that there are many other edges present within the sphere, or crossing its boundary, that we can charge to.

\vspace{5pt}
\noindent {\bf{Charging Scheme Overview:}}
Fix an arbitrary node $u\in V$.
In order to bound the number of misclassified edges incident on $u$, our analysis tracks two quantities of interest.
The first is the total number of edges incident on $u$ that are classified incorrectly by the algorithm. 
 Recall that this quantity is denoted by $\text{disagree}_{\mathcal{C}}(u)$, and we refer to it as $u$'s {\em cost}.
The second is the total fractional disagreement of node $u$ as given by the relaxation, \ie, $D(u)$.
We refer to $D(u)$ as $u$'s {\em budget}.
Since we consider unweighted complete graphs $D(u)$ reduces to:
$ D(u) = \sum _{v:(u,v)\in E^+}d\left( u,v\right) + \sum _{v:(u,v)\in E^-} \left( 1-d\left( u,v\right) \right)$.

Note that both $u$'s cost and budget are fixed.
However, it will be conceptually helpful to view these quantities as changing as the algorithm progresses.
Initially: $(1)$ $u$'s cost is $0$ as no edge has been classified yet, \ie, $ \text{disagree}_{\mathcal{C}}(u)=0$ once Algorithm \ref{alg:7ApproxClique} starts, and  $(2)$ $u$'s budget is full, \ie, $D(u) = \sum _{v:(u,v)\in E^+}d\left( u,v\right) + \sum _{v:(u,v)\in E^-} \left( 1-d\left( u,v\right) \right)$ once Algorithm \ref{alg:7ApproxClique} starts.
In every iteration $u$'s cost increases by the number of newly misclassified edges incident on $u$, and $u$'s budget decreases by the total fractional contribution of all newly classified edges incident on $u$ (whether correct or not).
Our analysis bounds the ratio of these two changes in each iteration of Algorithm \ref{alg:7ApproxClique}.

\begin{figure}
    \centering
    \begin{subfigure}[b]{0.49\textwidth}
        \includegraphics[width=\textwidth]{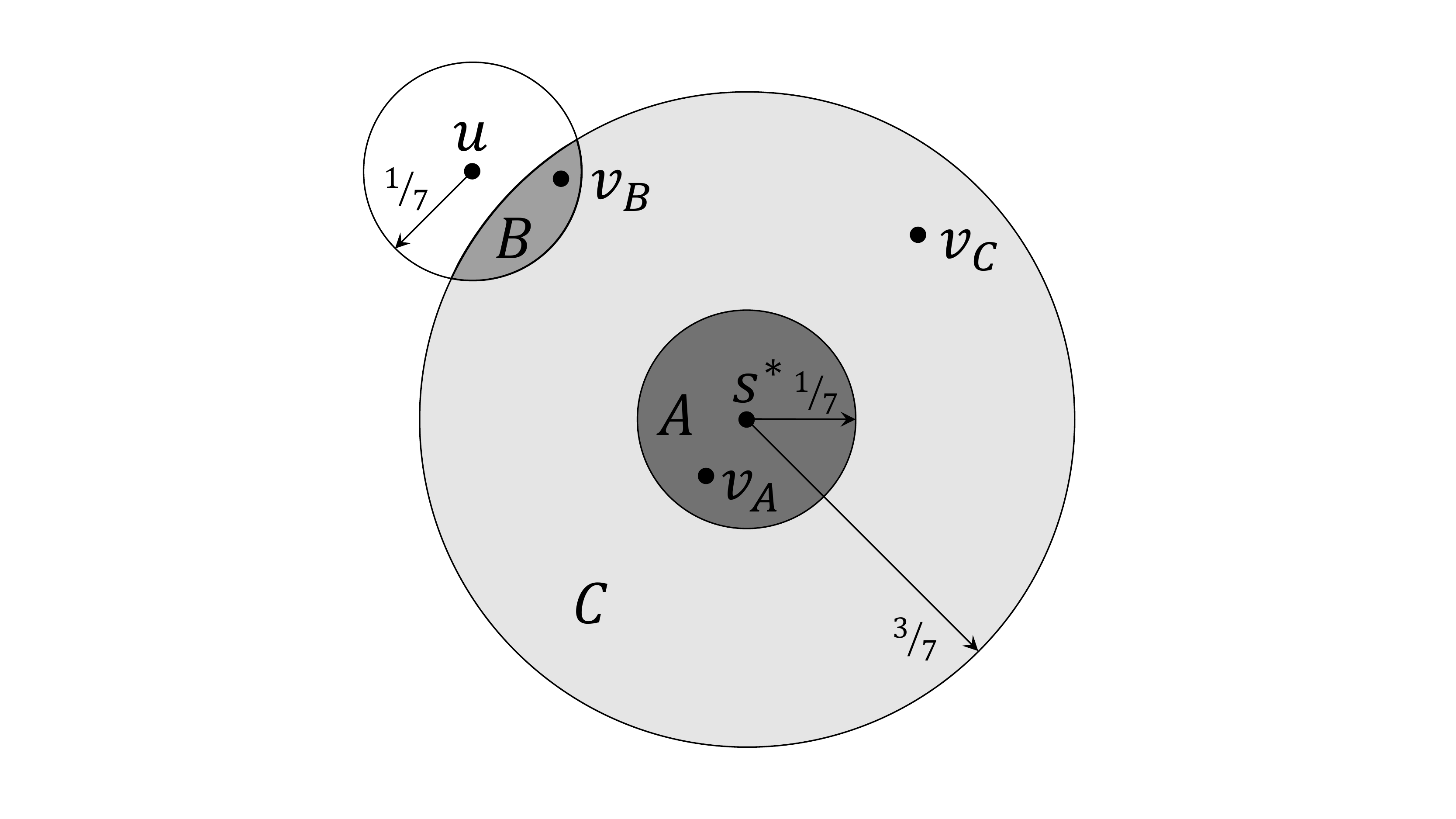}
        \caption{Case $1$: $u\notin \text{Ball}_S(s^*,\nicefrac[]{3}{7})$}
        \label{fig:Case1}
    \end{subfigure}
    \begin{subfigure}[b]{0.49\textwidth}
        \includegraphics[width=\textwidth]{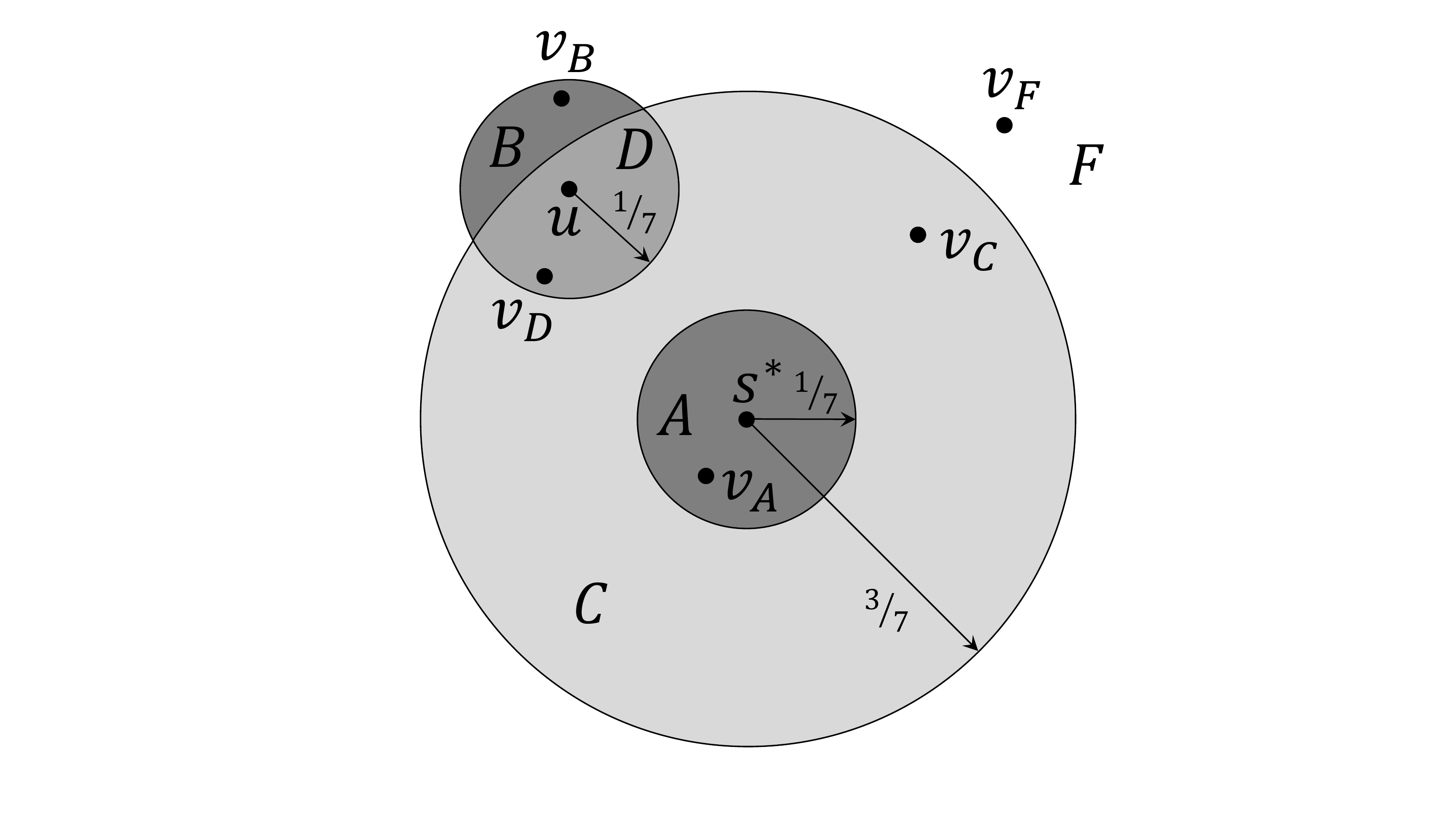}
        \caption{Case $2$: $u\in \text{Ball}_S(s^*,\nicefrac[]{3}{7})$}
        \label{fig:Case2}
    \end{subfigure}
    \caption{Cases of Lemmas \ref{lem:7ApproxClique} and \ref{lem:7ApproxBipartite} analysis.}\label{fig:Cases}
\end{figure}

\begin{proof}[of Lemma \ref{lem:7ApproxClique}] 
Fix a vertex $u$ and an arbitrary iteration.
Consider two cases depending on whether $u$ belongs to the cluster formed in the chosen iteration.
It is important to note that once $u$ is assigned to a cluster that is added to $\mathcal{C}$, its value, \ie, $\text{disagree}_{\mathcal{C}}(u)$, does not change in subsequent iterations and remains fixed until the algorithm terminates.

\vspace{5pt}
\noindent {\bf{Case $1$: $u\notin \text{Ball}_S(s^*,\nicefrac[]{3}{7})$:}}
Note that the only edges incident on $u$ that are classified incorrectly in the current iteration, are edges $(u,v)\in E^+$ for some node $v\in \text{Ball}_S(s^*,\nicefrac[]{3}{7})$.
Let us define the following disjoint collections of vertices: $A\triangleq \text{Ball}_S(s^*,\nicefrac[]{1}{7})$, $B\triangleq \text{Ball}_S(s^*,\nicefrac[]{3}{7})\cap \text{Ball}_S(u,\nicefrac[]{1}{7})$, and $C\triangleq \text{Ball}_S(s^*,\nicefrac[]{3}{7})\setminus \left( A \cup B\right)$.
Refer to Figure \ref{fig:Case1} for a drawing of $A$, $B$ and $C$.
Thus, the erroneously classified edges are $(u,v_A)\in E^+$ where $v_A\in A$, $(u,v_B)\in E^+$ where $v_B\in B$, and $(u,v_C)\in E^+$ where $v_C\in C$.


First, let us focus on edges $(u,v_C)\in E^+$.
Each edge $(u,v_C)$ increases $u$'s cost by $1$.
We charge this increase to the fractional contribution of $(u,v_C)$ to $u$'s budget, which equals $d(u,v_C)$.
Since $v_C\notin \text{Ball}_S(u,\nicefrac[]{1}{7})$ it must be the case that $d(u,v_C)\geq \nicefrac[]{1}{7}$.
Therefore, each $(u,v_C)$ edge incurs a multiplicative loss of at most $7$.

Let us focus now on edges $(u,v_B)\in E^+$ and $(u,v_A)$ simultaneously.
Since $s^*$ was chosen greedily, \ie, it maximizes the number of nodes within distance less than $\nicefrac[]{1}{7}$ from it, we can conclude that $|B| \leq |A|$.
Thus, each node in $B$ can be assigned to a {\em distinct} node in $A$.
Fix $v_B\in B$ and let $v_A\in A$ be the node assigned to it.

\begin{enumerate}
\item If $(u,v_A)\in E^+$ then the {\em joint} contribution of $(u,v_B)$ and $(u,v_A)$ to $u$'s cost is $2$.
We charge this cost to the fractional contribution of $(u,v_A)$ alone to $u$'s budget, which equals $d(u,v_A)$.
The triangle inequality implies that $d(u,v_A)\geq d(u,s^*) - d(s^*,v_A)\geq \nicefrac[]{3}{7} - \nicefrac[]{1}{7}= \nicefrac[]{2}{7}$.
Hence, both $(u,v_B)$ and $(u,v_A)$ incur a multiplicative loss of at most $\nicefrac[]{2}{\left(\nicefrac[]{2}{7}\right)}=7$.

\item If $(u,v_A)\in E^-$ then $(u,v_A)$ does not increase $u$'s cost, and therefore the increase in $u$'s cost is caused solely by $(u,v_B)$ and it equals $1$.
We charge this cost to the fractional contribution of $(u,v_A)$ alone to $u$'s budget, which equals $1-d(u,v_A)$.
The triangle inequality implies that $d(u,v_A)\leq d(u,v_B) + d(v_B,s^*)+ d(s^*,v_A)\leq \nicefrac[]{1}{7}+\nicefrac[]{3}{7}+\nicefrac[]{1}{7} = \nicefrac[]{5}{7}$.
Hence, $(u,v_B)$ incurs a multiplicative loss of at most $\frac{1}{1-\nicefrac[]{5}{7}}=\nicefrac[]{7}{2}$.

\item If there are any remaining nodes $v_A\in A$ that no node in $B$ was assigned to them, and $(u,v_A)\in E^+$, then we charge the $1$ cost $ (u,v_A)$ adds to $u$'s cost to the fractional contribution of $(u,v_A)$ to $u$'s budget, which equals $d(u,v_A)$.
The triangle inequality implies that $d(u,v_A)\geq d(u,s^*) - d(s^*,v_A)\geq \nicefrac[]{3}{7} - \nicefrac[]{1}{7}= \nicefrac[]{2}{7}$.
Hence, such an edge $(u,v_A)$ incurs a multiplicative loss of at most $\nicefrac[]{1}{\left(\nicefrac[]{2}{7}\right)}=\nicefrac[]{7}{2}$.
\end{enumerate}

\noindent Thus, we can conclude that for the first case in which $u\notin \text{Ball}_S(s^*,\nicefrac[]{3}{7})$ we lose a factor of at most $7$.

\vspace{5pt}
\noindent {\bf Case $2$: $u\in \text{Ball}_S(s^*,\nicefrac[]{3}{7})$:}
Note that the only edges incident on $u$ that are classified incorrectly in the current iteration, are edges $(u,v)\in E^+$ for some $v\notin \text{Ball}_S(s^*,\nicefrac[]{3}{7})$ and edges $(u,v)\in E^-$ for some $v\in \text{Ball}_S(s^*,\nicefrac[]{3}{7})$.
Let us define the following disjoint collections of vertices: $A\triangleq \text{Ball}_S(s^*,\nicefrac[]{1}{7})$, $B\triangleq \text{Ball}_S(u,\nicefrac[]{1}{7})\setminus \text{Ball}_S(s^*,\nicefrac[]{3}{7})$, $C\triangleq \text{Ball}_S(s^*,\nicefrac[]{3}{7})\setminus \left( A \cup \text{Ball}_S(u,\nicefrac[]{1}{7})\right)$, $D\triangleq \text{Ball}_S(u,\nicefrac[]{1}{7})\cap \text{Ball}_S(s^*,\nicefrac[]{3}{7})$, and $F\triangleq S \setminus (A\cup B \cup C \cup D)$. Refer to Figure \ref{fig:Case2} for a drawing of $A$, $B$, $C$, $D$, and $F$.
Thus, the erroneously classified edges are $(u,v_A)\in E^-$ where $v_A\in A$, $(u,v_B)\in E^+$ where $v_B\in B$, $(u,v_C)\in E^-$ where $v_C\in C$, $(u,v_D)\in E^-$ where $v_D\in D$, and $(u,v_F)\in E^+$ where $v_F\in F$.

Let us focus now on edges $(u,v_B)\in E^+$ and $(u,v_A)$ simultaneously.
Since $s^*$ was chosen greedily, \ie, it maximizes the number of nodes within distance less than $\nicefrac[]{1}{7}$ from it, we can conclude that $|B| \leq |A|$.
Thus, each node in $B$ is assigned to a {\em distinct} node in $A$.
Fix $v_B\in B$ and let $v_A\in A$ be the node assigned to it.
\begin{enumerate}
\item If $(u,v_A)\in E^+$ then $(u,v_A)$ does not increase $u$'s cost, and therefore the increase in $u$'s cost is caused solely by $(u,v_B)$ and it equals $1$.
We charge this cost to the fractional contribution of $(u,v_A)$ alone to $u$'s budget, which equals $d(u,v_A)$.
The triangle inequality implies that $d(v_A,u)\geq d(s^*,v_B) -  d(s^*,v_A) - d(u,v_B)\geq \nicefrac[]{3}{7}-\nicefrac[]{1}{7}-\nicefrac[]{1}{7} = \nicefrac[]{1}{7}$.
Hence, $(u,v_B)$ incurs a multiplicative loss of at most $\nicefrac[]{1}{\left(\nicefrac[]{1}{7}\right)}=7$.

\item If $(u,v_A)\in E^-$ then the {\em joint} contribution of $(u,v_B)$ and $(u,v_A)$ to $u$'s cost is $2$.
We charge this cost to the fractional contribution of $(u,v_A)$ alone to $u$'s budget, which equals $1-d(u,v_A)$.
The triangle inequality implies that $d(u,v_A)\leq d(u,s^*) + d(s^*,v_A)\leq \nicefrac[]{3}{7} + \nicefrac[]{1}{7}= \nicefrac[]{4}{7}$.
Hence, both $(u,v_B)$ and $(u,v_A)$ incur a multiplicative loss of at most $\nicefrac[]{2}{(1-\nicefrac[]{4}{7})}=\nicefrac[]{14}{3}$.

\item If there are any remaining nodes $v_A\in A$ that no node in $B$ was assigned to them, and that $(u,v_A)\in E^-$, we charge the $1$ cost $ (u,v_A)$ adds to $u$'s cost to the fractional contribution of $(u,v_A)$ to $u$'s budget, which equals $1-d(u,v_A)$.
As before, the triangle inequality implies that $d(u,v_A)\leq d(u,s^*) + d(s^*,v_A)\leq \nicefrac[]{3}{7} + \nicefrac[]{1}{7}= \nicefrac[]{4}{7}$.
Hence, such an edge $(u,v_A)$ incurs a multiplicative loss of at most $\nicefrac[]{1}{(1-\nicefrac[]{4}{7})}=\nicefrac[]{7}{3}$.
\end{enumerate}

Let us now focus on edges $(u,v_C)\in E^-$ and $(u,v_D)\in E^-$.
For simplicity, let us denote such an edge by $(u,v)$ where $v\in C\cup D$.
Each such edge increases $u$'s cost by $1$.
We charge this increase to the fractional contribution of the same edge to $u$'s budget, which equals $1-d(u,v)$.
Since both $u,v\in \text{Ball}_S(s^*,\nicefrac[]{3}{7})$ the triangle inequality implies that $d(u,v)\leq d(u,s^*) + d(s^*,v)\leq \nicefrac[]{6}{7}$.
Therefore, each $(u,v)\in E^-$, where $ v\in C\cup D$, incurs a multiplicative loss of at most $\nicefrac[]{1}{(1-\nicefrac[]{6}{7})}=7$.

Finally, consider edges $(u,v_F)\in E^+$.
Each such edge increases $u$'s cost by 1.
We charge this increase to the fractional contribution of the same edge to $u$'s budget, which equals $d(u,v_F)$.
Since $d(u,v_F) \geq \nicefrac[]{1}{7}$, each such edge incurs a multiplicative loss of at most $\nicefrac[]{1}{\left(\nicefrac[]{1}{7}\right)}=7$.
This concludes the proof as we have shown that for every vertex $u$ and every iteration, the increase in $u$'s cost during the iteration as it most $7$ times the decrease in $u$'s budget during the same iteration. \hfill $\square$
\end{proof}

\section{Proof of Theorem \ref{thrm:7ApproxBipartiteMLD}}\label{app:7Bipartite}

Let $G=(V,E)$ be an unweighted complete bipartite graph. Let $V_1$ and $V_2$ be the two sides of the graph G. Our algorithm will ensure a $7$ approximation factor for mistakes on all vertices in $V_1$ but does not give any guarantee for vertices in $V_2$. The algorithm is a slight modification of the Algorithm \ref{alg:7ApproxClique} presented earlier.

We consider the following simple deterministic greedy clustering algorithm.
Algorithm \ref{alg:7ApproxBipartite} receives as input the metric $d$ (as computed by the relaxation (\ref{Relaxation:Disagreements})), whereas the variables $D(u)$ are required only for the analysis.
In every step, the algorithm greedily chooses a vertex $s^* \in V_1$ that has many vertices in $V_2$ {\em close} to it with respect to the metric $d$.
Then, $s^*$ just cuts a {\em large} sphere around it to form a new cluster.
\begin{algorithm}
\caption{Greedy Clustering $\left( \left\{ d(u,v)\right\} _{u,v\in V}\right)$}\label{alg:7ApproxBipartite}
\begin{algorithmic}[1]
\STATE $S\leftarrow V$ and $\mathcal{C}\leftarrow \emptyset$.
\WHILE {$S\cap V_1\neq \emptyset$}
\STATE $s^*\leftarrow \text{argmax}\left\{ \left| \text{Ball}_{V_2}(s,\nicefrac[]{1}{7})\right|:s\in V_1\right\}$.
\STATE $\mathcal{C} ~\leftarrow \mathcal{C} \cup \left\{ \text{Ball}_S(s^*,\nicefrac[]{3}{7})\right\}$.
\STATE $S~\leftarrow S\setminus \text{Ball}_S(s^*,\nicefrac[]{3}{7})$.
\ENDWHILE
\WHILE {$S\neq \emptyset$}
\STATE $s^*\leftarrow s^* \in S $.
\STATE $\mathcal{C} ~\leftarrow \mathcal{C} \cup \left\{ s^* \right\}$.
\STATE $S~\leftarrow S\setminus s^*$.
\ENDWHILE
\STATE Output $\mathcal{C}$.
\end{algorithmic}
\end{algorithm}

The following lemma summarizes the guarantee achieved by Algorithm \ref{alg:7ApproxBipartite}.
\begin{lemma}\label{lem:7ApproxBipartite}
Assuming the input is an unweighted complete bipartite graph, Algorithm \ref{alg:7ApproxBipartite} guarantees that $\text{disagree}_{\mathcal{C}}(u) \leq 7D(u)$ for any $u\in V_1$.
\end{lemma}

\noindent {\bf{Charging Scheme Overview:}}
Fix an arbitrary vertex $u\in V_1$.
As before, we track two quantities: $u$'s cost and $u$'s budget.
Our analysis bounds the ratio of the change in these two quantities in each iteration of Algorithm \ref{alg:7ApproxBipartite}.
%

\begin{proof}[of Lemma \ref{lem:7ApproxBipartite}]
Fix a vertex $u \in V_1$ and an arbitrary iteration.
We consider two cases depending on whether $u$ belongs to the cluster formed in the chosen iteration.
It is important to note that once $u$ is chosen to a cluster that is added to $\mathcal{C}$, its value, \ie, $\text{disagree}_{\mathcal{C}}(u)$, does not change and remains fixed until the algorithm terminates.

\vspace{5pt}
\noindent {\bf Case $1$: $u\notin \text{Ball}_S(s^*,\nicefrac[]{3}{7}), u \in V_1$:}
Note that the only edges incident on $u$ that are classified incorrectly in the
current iteration, are edges $(u,v)\in E^+$ for some $v\in \text{Ball}_S(s^*,\nicefrac[]{3}{7}) \cap V_2$.
Define the following disjoint collections of vertices: $A\triangleq \text{Ball}_S(s^*,\nicefrac[]{1}{7}) \cap V_2$, $B\triangleq \text{Ball}_S(s^*,\nicefrac[]{3}{7})\cap \text{Ball}_S(u,\nicefrac[]{1}{7}) \cap V_2$, and $C\triangleq (\text{Ball}_S(s^*,\nicefrac[]{3}{7}) \cap V_2)\setminus \left( A \cup B\right)$. Refer to Figure \ref{fig:Case1} for a drawing of $A$, $B$ and $C$.
Thus, the erroneously classified edges are $(u,v_A)\in E^+$ where $v_A\in A$, $(u,v_B)\in E^+$ where $v_B\in B$, and $(u,v_C)\in E^+$ where $v_C\in C$. Note that whenever there is an edge $(u,v)$ where $u \in V_1$, $v$ must belong to $V_2$ since G is a bipartite graph.

Let us focus on edges $(u,v_C)\in E^+$.
Each edge $(u,v_C)$ increases $u$'s cost by $1$.
We charge this increase to the fractional contribution of $(u,v_C)$ to $u$'s budget, which equals $d(u,v_C)$.
Since $v_C\notin \text{Ball}_S(u,\nicefrac[]{1}{7})$ it must be the case that $d(u,v_C)\geq \nicefrac[]{1}{7}$.
Therefore, each $(u,v_C)$ edge incurs a multiplicative loss of at most $7$.

Let us focus now on edges $(u,v_A)$ and $(u,v_B) \in E^+$ simultaneously.
Since $s^*$ was chosen greedily, \ie, it maximizes the number of nodes $\in V_2$ within distance of at most $\nicefrac[]{1}{7}$ from it, we can conclude that $|B| \leq |A|$.
Thus, each node in $B$ can be assigned to a {\em distinct} node in $A$.
Fix $v_B\in B$ and let $v_A\in A$ be the node assigned to it.
\begin{enumerate}
\item If $(u,v_A)\in E^+$ then the {\em joint} contribution of $(u,v_B)$ and $(u,v_A)$ to $u$'s cost is $2$.
We charge this cost to the fractional contribution of $(u,v_A)$ alone to $u$'s budget, which equals $d(u,v_A)$.
The triangle inequality implies that $d(u,v_A)\geq d(u,s^*) - d(s^*,v_A)\geq \nicefrac[]{3}{7} - \nicefrac[]{1}{7}= \nicefrac[]{2}{7}$.
Hence, both $(u,v_B)$ and $(u,v_A)$ incur a multiplicative loss of at most $\nicefrac[]{2}{\left(\nicefrac[]{2}{7}\right)}=7$.
\item If $(u,v_A)\in E^-$ then $(u,v_A)$ does not increase $u$'s cost, and therefore the increase in $u$'s cost is caused solely by $(u,v_B)$ and it equals $1$.
    We charge this cost to the fractional contribution of $(u,v_A)$ alone to $u$'s budget, which equals $1-d(u,v_A)$.
    The triangle inequality implies that $d(u,v_A)\leq d(u,v_B) + d(v_B,s^*)+ d(s^*,v_A)\leq \nicefrac[]{1}{7}+\nicefrac[]{3}{7}+\nicefrac[]{1}{7} = \nicefrac[]{5}{7}$.
    Hence, $(u,v_B)$ incurs a multiplicative loss of at most $\frac{1}{1-\nicefrac[]{5}{7}}=\nicefrac[]{7}{2}$.
\item If there are any remaining nodes $v_A\in A$ such no node in $B$ was assigned to them, and $(u,v_A)\in E^+$, we charge the $1$ cost $ (u,v_A)$ adds to $u$'s cost to the fractional contribution of $(u,v_A)$ to $u$'s budget, which equals $d(u,v_A)$.
    The triangle inequality implies that $d(u,v_A)\geq d(u,s^*) - d(s^*,v_A)\geq \nicefrac[]{3}{7} - \nicefrac[]{1}{7}= \nicefrac[]{2}{7}$.
    Hence, such an edge $(u,v_A)$ incurs a multiplicative loss of at most $\nicefrac[]{1}{\left(\nicefrac[]{2}{7}\right)}=\nicefrac[]{7}{2}$.
\end{enumerate}
We can conclude that the first case in which $u\notin \text{Ball}_S(s^*,\nicefrac[]{3}{7})$ we lose a factor of at most $7$.

\vspace{5pt}
\noindent {\bf Case $2$:  $u\in \text{Ball}_S(s^*,\nicefrac[]{3}{7}) \cap V_1$: }
Note that the only edges incident on $u$ that are classified incorrectly in the current iteration, are edges $(u,v)\in E^+$ for some $v\notin \text{Ball}_S(s^*,\nicefrac[]{3}{7}), v \in V_2$ and edges $(u,v)\in E^-$ for some $v\in \text{Ball}_S(s^*,\nicefrac[]{3}{7}) \cap V_2$.
Define the following disjoint collections of vertices: $A\triangleq \text{Ball}_S(s^*,\nicefrac[]{1}{7}) \cap V_2$, $B\triangleq (\text{Ball}_S(u,\nicefrac[]{1}{7}) \cap V_2)\setminus \text{Ball}_S(s^*,\nicefrac[]{3}{7})$, $C\triangleq (\text{Ball}_S(s^*,\nicefrac[]{3}{7}) \cap V_2)\setminus \left( A \cup \text{Ball}_S(u,\nicefrac[]{1}{7})\right)$, $D\triangleq \text{Ball}_S(u,\nicefrac[]{1}{7})\cap \text{Ball}_S(s^*,\nicefrac[]{3}{7}) \cap V_2$ and $F\triangleq (V_2 \cup S) \setminus (A\cup B \cup C \cup D)$. Refer to Figure \ref{fig:Case2} for a drawing of $A$, $B$, $C$, $D$ and $F$.
Thus, the erroneously classified edges are $(u,v_A)\in E^-$ where $v_A\in A$, $(u,v_B)\in E^+$ where $v_B\in B$, $(u,v_C)\in E^-$ where $v_C\in C$, $(u,v_D)\in E^-$ where $v_D\in D$ and $(u,v_F)\in E^+$ where $v_F\in F$.

Let us focus now on edges $(u,v_A)$ and $(u,v_B) \in E^+$ simultaneously.
Since $s^*$ was chosen greedily, \ie, it maximizes the number of nodes $\in V_2$ within distance of at most $\nicefrac[]{1}{7}$ from it, we can conclude that $|B| \leq |A|$.
Thus, each node in $B$ can be assigned to a {\em distinct} node in $A$.
Fix $v_B\in B$ and let $v_A\in A$ be the node assigned to it.
\begin{enumerate}
\item If $(u,v_A)\in E^-$ then the {\em joint} contribution of $(u,v_B)$ and $(u,v_A)$ to $u$'s cost is $2$.
We charge this cost to the fractional contribution of $(u,v_A)$ alone to $u$'s budget, which equals $1-d(u,v_A)$.
The triangle inequality implies that $d(u,v_A)\leq d(u,s^*) + d(s^*,v_A)\leq \nicefrac[]{3}{7} + \nicefrac[]{1}{7}= \nicefrac[]{4}{7}$.
Hence, both $(u,v_B)$ and $(u,v_A)$ incur a multiplicative loss of at most $\nicefrac[]{2}{\left(1-\nicefrac[]{4}{7}\right)}=\nicefrac[]{14}{3}$.
\item If $(u,v_A)\in E^+$ then $(u,v_A)$ does not increase $u$'s cost, and therefore the increase in $u$'s cost is caused solely by $(u,v_B)$ and it equals $1$.
    We charge this cost to the fractional contribution of $(u,v_A)$ alone to $u$'s budget, which equals $d(u,v_A)$.
    The triangle inequality implies that $d(u,v_A)\geq d(u,s^*) - d(v_A,s^*) \geq d(v_B,s^*) - d(v_B,u) - d(v_A,s^*)\geq \nicefrac[]{3}{7}-\nicefrac[]{1}{7}-\nicefrac[]{1}{7} = \nicefrac[]{1}{7}$.
    Hence, $(u,v_B)$ incurs a multiplicative loss of at most $\nicefrac[]{1}{\left(\nicefrac[]{1}{7}\right)}=7$.
\item If there are any remaining nodes $v_A\in A$ such that no node in $B$ was assigned to them, and $(u,v_A)\in E^-$, then $u$'s cost increases by 1. We charge this cost to the fractional contribution of $(u,v_A)$ to $u$'s budget, which equals $1-d(u,v_A)$.
    The triangle inequality implies that $d(u,v_A)\leq d(u,s^*) + d(s^*,v_A)\leq \nicefrac[]{3}{7} + \nicefrac[]{1}{7}= \nicefrac[]{4}{7}$.
    Hence, such an edge $(u,v_A)$ incurs a multiplicative loss of at most $\nicefrac[]{1}{\left(1-\nicefrac[]{4}{7}\right)}=\nicefrac[]{7}{3}$.
\end{enumerate}

Let us focus on edges $(u,v_C)\in E^-$, and $(u,v_D)\in E^-$.
For simplicity, let us denote such an edge by $(u,v)$ where $v\in C\cup D$.
Each such edge increases $u$'s cost by $1$.
We charge this increase to the fractional contribution of the same edge to $u$'s budget, which equals $1-d(u,v)$.
Since both $u,v\in \text{Ball}_S(s^*,\nicefrac[]{3}{7})$ it must be the case from the triangle inequality that $d(u,v)\leq d(u,s^*) + d(s^*,v)\leq \nicefrac[]{6}{7}$.
Therefore, each $(u,v)\in E^-$, where $ v\in C\cup D$, incurs a multiplicative loss of at most $\nicefrac[]{1}{(1-\nicefrac[]{6}{7})}=7$.

Now, let us consider edges $(u,v_F)\in E^+$. Each such edge increases $u$'s cost by 1. We charge this increase to the fractional contribution of the same edge to $u$'s budget, which equals $d(u,v_F)$. Since $d(u,v_F) \geq \nicefrac[]{1}{7}$. Hence, each such edge incurs a multiplicative loss of at most $\nicefrac[]{1}{\nicefrac[]{1}{7}}=7$. This concludes the proof as we have shown that for every vertex $u \in V_1$ and every iteration, the increase in $u$'s cost during the iteration as it most 7 times the decrease in $u$'s budget during the same iteration. \hfill $\square$
\end{proof}

\noindent We are now ready to prove Theorem \ref{thrm:7ApproxBipartiteMLD}.

\begin{proof}[of Theorem \ref{thrm:7ApproxBipartiteMLD}]
Apply Algorithm \ref{alg:7ApproxBipartite} to the solution of the relaxation (\ref{Relaxation:Disagreements}).
Lemma \ref{lem:7ApproxBipartite} guarantees that for every node $u\in V_1$ we have that $ \text{disagree}_{\mathcal{C}}(u) \leq 7D(u)$, \ie, $ \text{disagree}_{\mathcal{C}}(V_1)\leq 7\bD$.\footnote{For simplicity we denote here by $\bD$ the vector of $D(u)$ variables for vertices in $ V_1$.}
The value of the output of the algorithm is $f\left( \text{disagree}_{\mathcal{C}}(V_1)\right)$ and one can bound it as follows:
$$ f\left( \text{disagree}_{\mathcal{C}}(V_1)\right)\stackrel{(1)}{\leq} f\left( 7\bD\right) \stackrel{(2)}{\leq} 7f\left( \bD\right)~.$$
Inequality $(1)$ follows from the monotonicity of $f$, whereas inequality $(2)$ follows from the scaling property of $f$.
This concludes the proof since $f\left( \bD\right) $ is a lower bound on the value of any optimal solution.\hfill $\square$
\end{proof}

\section{Proof of Theorem \ref{thrm:ApproxMMA}}\label{app:ApproxMMA}
For completeness and intuition, we start with an exposition on the simple local search algorithm for {\MLA}.
Let us denote by $c(u)$ the total weight of edges incident on $u$, and by $\bc\in \mathcal{R}^V$ the vector of all $\{c(u)\}_{u\in V}$.
Additionally, for any cut $S\subseteq V$ we denote by $\mathcal{C}_S=\{ S,\bar{S}\}$ the clustering $S$ defines.
The simple local search algorithm starts with an arbitrary cut $S$, and repeatedly moves vertices from one side to the other until no additional improvement can be made.
Specifically, if $\text{agree}_{\mathcal{C}_S}(u) < \nicefrac[]{c(u)}{2} $ then node $u$ is moved to the other side of the cut.
%

When the algorithm terminates it must be that for every node $u$: $\text{agree}_{\mathcal{C}_S}(u) \geq \nicefrac[]{c(u)}{2}$.
The latter implies that $\mathcal{C}_S$ is a $\nicefrac[]{1}{2}$-approximation for {\MLA} since:
$$ g(\text{agree}_{\mathcal{C}_S}(V))\stackrel{(1)}{\geq} g(\nicefrac[]{\bc}{2}) \stackrel{(2)}{\geq} \nicefrac[]{g(\bc)}{2}~.$$
Inequality $(1)$ follows from the monotonicity of $g$, whereas inequality $(2)$ follows from the reverse scaling property of $g$.
Note that $g(\bc)$ upper bounds the value of any optimal solution to {\MLA}.

One can track the progress of the algorithm by considering the potential $\Phi_{\mathcal{C}_S} \triangleq \sum _{u\in V}\text{agree}_{\mathcal{C}_S}(u)$.
For every node $u$ and cut $S$ note that: $(1)$ $\text{agree}_{\mathcal{C}_S}(u) + \text{disagree}_{\mathcal{C}_S}(u) = c(u)$, and $(2)$
if $u$ is moved to the other side of the cut then the values of $ \text{agree}_{\mathcal{C}_S}(u)$ and $\text{disagree}_{\mathcal{C}_S}(u)$ are swapped.
Thus, the potential $ \Phi_{\mathcal{C}_S}$ must strictly increase in every iteration, implying that the algorithm always terminates.
Unfortunately, it is well known that in general the local search algorithm might terminate after an exponential number of iterations.


For {\MMA} we are able to show that a {\em non-oblivious} local search succeeds in finding a $\nicefrac[]{1}{(2+\varepsilon)}$-approximation in polynomial time, thus matching the existential guarantee of the simple local search algorithm.
Our main idea is to alter the edge weights in such a way that in the new resulting instance the ratio of $\max _{S\subseteq V}\left\{ \Phi_{\mathcal{C}_S}\right\}$ to the value of the optimal solution for {\MMA} is polynomially bounded.
Once this is guaranteed, the local search algorithm can be altered so it always terminates in polynomial time.
We are now ready to prove Theorem \ref{thrm:ApproxMMA}.

\begin{proof}[of Theorem \ref{thrm:ApproxMMA}]
First we describe the process of creating the new edge weights.
Let $c^*\triangleq \min _{u\in V}\left\{ c(u)\right\}$ be the minimum total weight of edges incident on any node.
Clearly, $c^*$ serves as an upper bound on the value of any optimal solution for {\MMA}.

Denote by $T\triangleq \left\{ u\in V:c(u)=c^*\right\}$ the collection of all nodes whose total weight of edges incident on each of them is $c^*$.
We are going to describe a process that only decreases edge weights, and we prove that once it terminates all the following are true: $(1)$ $c^*$ does not decrease, $(2)$ $T$ still contains exactly all nodes whose total weight of edges incident on each of them is $c^*$, and $(3)$ $E(\overline{T})=\emptyset$, \ie, there is no edge $(u,v)$ such that both $u,v\notin T$.

The process is defined as follows.
While there is an edge $(u,v)\in E(\overline{T})$, \ie, both $u,v\notin T$, whose weight is $c_{u,v}>0$, decrease its weight until the first of the following happens: $c_{u,v}$ reaches $0$ (in which case we remove the edge), $c(u)$ reaches $c^*$ (in which case we add $u$ to $T$ and stop decreasing the weight of the edge), or $c(v)$ reaches $c^*$ (in which case we add $v$ to $T$ and stop decreasing the weight of the edge).
Clearly this process terminates in polynomial time, and $(1)$, $(2)$, and $(3)$ above are all satisfied (see Figure \ref{fig:NoEdgesT}).

\begin{figure}
    \centering
    \includegraphics[width=0.75\textwidth]{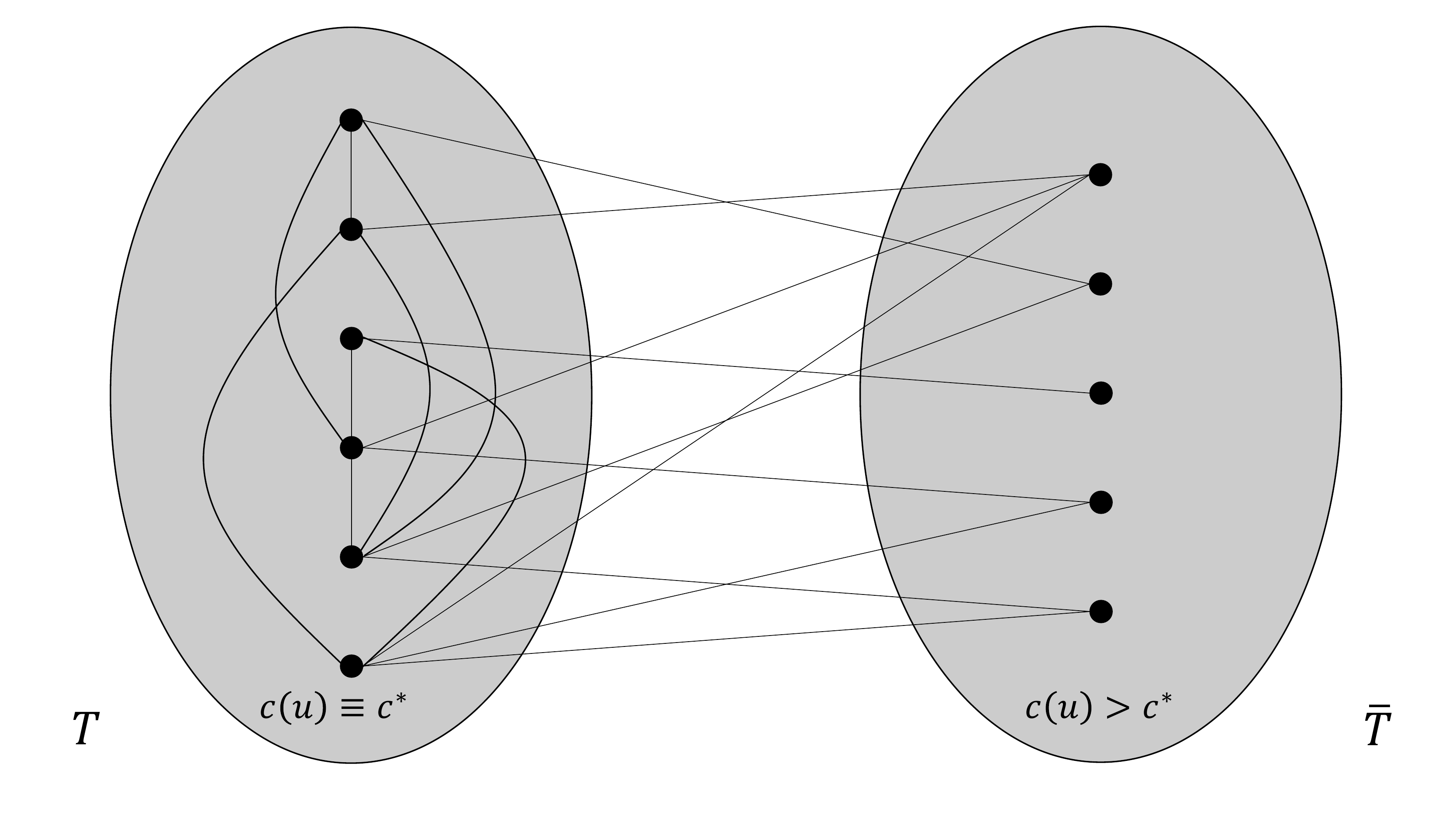}
    \caption{No edges in $E(\overline{T})$ after creating new weights.}
    \label{fig:NoEdgesT}
\end{figure}

We now execute the following modified local search algorithm on the new graph $G$ and edge weight function.
Its full description is given by Algorithm \ref{alg:MMALocalSearch}.
\begin{algorithm}
\caption{Non-Oblivious Local Search ($G=(V,E),c^*,\varepsilon$)}\label{alg:MMALocalSearch}
\begin{algorithmic}[1]
\STATE $i\leftarrow 0$ and choose an arbitrary $S_0\subseteq V$.
\WHILE {$\exists u \in V \text{such that } \text{agree}_{\mathcal{C}_{S_i}}(u) < (\nicefrac[]{1}{2}-\varepsilon)c^*$}
\STATE $\text{move }u \text{ to the other side of the cut } S_i$ and denote the resulting cut by $S_{i+1}$.
\STATE $i\leftarrow i + 1$.
\ENDWHILE
\STATE output $ \mathcal{C}_{S_i}$.
\end{algorithmic}
\end{algorithm}

Clearly, once Algorithm \ref{alg:MMALocalSearch} terminates: $\text{agree}_{\mathcal{C}_{S_i}}(u)\geq (\nicefrac[]{1}{2}-\varepsilon)c^*$ for every node $ u\in V$.
Thus, the output is a $\nicefrac[]{1}{(2+\varepsilon')}$-approximation for {\MMA}, where $\varepsilon'=\nicefrac[]{4\varepsilon}{(1-2\varepsilon)}$.
All that remains is to prove that Algorithm \ref{alg:MMALocalSearch} terminates after a polynomial number of iterations.

For any $S\subseteq V$ define the potential $\Phi_{\mathcal{C}_S} \triangleq \sum _{u\in V}\text{agree}_{\mathcal{C}_S}(u)$ as before.
Note that:
\begin{align}
\max _{S\subseteq V}\left\{ \Phi_{\mathcal{C}_S}\right\} \stackrel{(1)}{\leq} 2\sum _{e\in E}c(e) \stackrel{(2)}{\leq} 2\sum _{u\in T}c(u) \stackrel{(3)}{\leq} 2n c^*~.\label{PotentialBound}
\end{align}
Inequality $(1)$ follows from the observation that $\text{agree}_{\mathcal{C}_S}(u)\leq c(u)$ for every $u\in V$, and thus the total potential can never exceed twice the total weight of edges in the graph.
Inequality $(2)$ follows from the fact that $E(\overline{T})=\emptyset$, whereas inequality $(3)$ follows from the definition of $T$ and the fact that $ |T|\leq n$.
Therefore, we can conclude that the potential $\Phi_{\mathcal{C}_S}$ is upper bounded by $2nc^*$.

Now we claim that in every iteration of Algorithm \ref{alg:MMALocalSearch} the potential $ \Phi_{\mathcal{C}_S}$ must increase by at least $2\varepsilon c^*$.
Fix an iteration $i$ and let $u$ be the node that was moved in this iteration.
Note that:
\begin{align}
\Phi _{\mathcal{C}_{S_{i+1}}} - \Phi _{\mathcal{C}_{S_{i}}} & \stackrel{(4)}{=} 2\left( \text{agree}_{\mathcal{C}_{S_{i+1}}}(u) - \text{agree}_{\mathcal{C}_{S_{i}}}(u)\right) \nonumber\\
& \stackrel{(5)}{=} 2\left( c(u) - 2 \cdot \text{agree}_{\mathcal{C}_{S_{i}}}(u)\right) \nonumber\\
& \stackrel{(6)}{\geq} 4\varepsilon c^*~.\label{PotentialGain}
\end{align}
Equality $(4)$ follows from the definition of the potential.
Since it is always the case that $\text{agree}_{\mathcal{C}_{S_{i}}}(u)+\text{disagree}_{\mathcal{C}_{S_{i}}}(u) =c(u)$, and the values of $\text{agree}_{\mathcal{C}_{S_{i}}}(u)$ and $\text{disagree}_{\mathcal{C}_{S_{i}}}(u)$ are swapped once $u$ is moved, \ie, $ \text{agree}_{\mathcal{C}_{S_{i+1}}}(u)=\text{disagree}_{\mathcal{C}_{S_{i}}}(u)$, we can conclude that equality $(5)$ is true.
Inequality $(6)$ holds since $c(u)\geq c^*$ and the reason $u$ was moved in iteration $i$ is that $\text{agree}_{\mathcal{C}_{S_{i}}}(u)<\left( \nicefrac[]{1}{2}-\varepsilon\right) c^*$.
Combining (\ref{PotentialBound}) and (\ref{PotentialGain}) proves that Algorithm \ref{alg:MMALocalSearch} terminates after at most $\nicefrac[]{n}{(2\varepsilon)}$ iterations. \hfill $\square$
\end{proof}

\section{Proof of Theorem \ref{thrm:IntegralityGapMMA}}\label{app:IntegralityGapMMA}
We prove now that the same integrality gap example used in proving Theorem \ref{thrm:IntegralityGapMMD} applies also for the current Theorem \ref{thrm:IntegralityGapMMA}.

\begin{proof}[of Theorem \ref{thrm:IntegralityGapMMA}]
Let $G$ be the unweighted cycle on $n$ vertices, where all edges are labeled $+$ and one edge is labeled $-$.
Specifically, denote the vertices of $G$ by $\left\{ v_1,v_2,\ldots,v_n\right\}$ where there is an edge $(v_i,v_{i+1})\in E^+$ for every $i=1,\ldots,n-1$ and additionally the edge $(v_n,v_1)\in E^-$.

First, we prove that the value of any integral solution is at most $1$.
A clustering that includes $V$ as a single cluster has value of $1$, as both $v_1$ and $v_n$ have exactly one correctly classified edge touching them i.e. 1 agreement.
Moreover, one can easily verify that any clustering into two or more clusters has a value of at most $1$.
Thus, any integral solution for the above instance has value of at most $1$.

Consider the natural linear programming relaxation for {\MMA}:
\begin{align*}
\max ~~~ & \min _{u\in V}\left\{ A(u)\right\} & \\
& \sum _{v:(u,v)\in E^+}c_{u,v}(1-d\left( u,v\right)) + \sum _{v:(u,v)\in E^-} c_{u,v}d\left( u,v \right) = A(u) & \forall u\in V \\
& d(u,v) + d(v,w) \geq d(u,w) & \forall u,v,w\in V \\
& A(u)\geq 0, ~0\leq d(u,v) \leq 1 & \forall u,v\in V
\end{align*}
Let us construct a fractional solution.
Assign a length of $\nicefrac[]{1}{n}$ for every $+$ edge and a length of $1-\nicefrac[]{1}{n}$ for the single $-$ edge, and let $d$ be the shortest path metric in $G$ induced by these lengths.
Obviously, the triangle inequality is satisfied and one can verify that $d(u,v)\leq 1$ for all $u,v\in V$.
Consider a vertex $v_i$ that does not touch the $-$ edge, \ie, $i=2,\ldots,n-1$.
Such a $v_i$ has two + edges touching it both having a length of $\nicefrac[]{1}{n}$, hence $A(v_i)=2-\nicefrac[]{2}{n}$.
Focusing on $v_1$ and $v_n$, each has one $+$ edge whose length is $\nicefrac[]{1}{n}$ and one $-$ edge whose length is $1-\nicefrac[]{1}{n}$ touching them.
Hence, $A(v_1)=A(v_n)=2-\nicefrac[]{2}{n}$.
Therefore, the above instance has an integrality gap of $\nicefrac[]{n}{(2(n-1))}$.

Now, consider the natural semi-definite relaxation for {\MMA}, where each vertex $u$ corresponds to a unit vector $\by _u$.
Intuitively, if $S_1,\ldots, S_{\ell}$ is an integral clustering, then all vertices in cluster $S_j$ are assigned to the standard $j$\textsuperscript{th} unit vector, \ie, $\be _j$.
Hence, the natural semi-definite relaxation requires that all vectors lie in the same orthant, \ie, for every $u$ and $v$: $\by _u \cdot \by _v\geq 0$, and that $\left\{ \by _u\right\} _{u\in V}$ satisfy the $\ell _2^2$ triangle inequality.
\begin{align*}
\max ~~~ & \min _{u\in V}\left\{ A(u)\right\} & \\
& \sum _{v:(u,v)\in E^+}c_{u,v}\left( \by _u \cdot \by _v\right) + \sum _{v:(u,v)\in E^-} c_{u,v}\left(1-\by _u \cdot \by _v\right) = A(u) & \forall u\in V \\
& || \by _u - \by _v||_2^2 + || \by _v - \by _w||_2^2 \geq || \by _u - \by _w||_2^2 & \forall u,v,w\in V \\
& \by _u \cdot \by _u = 1 & \forall u\in V \\
& \by _u \cdot \by _v \geq 0 & \forall u,v\in V
\end{align*}

In order to construct a fractional solution, it will be helpful to consider $Y\in \mathcal{R}^{V\times V}$ the positive semi-definite matrix of all inner products of $\left\{ \by _{v_i}\right\} _{i=1}^n$, \ie, $Y_{v_i,v_j}=\by _{v_i}\cdot \by _{v_j}$.
Intuitively, we consider a collection of integral solutions where for each one we construct the corresponding $Y$ matrix.
At the end, our fractional solution will be the average of all these $Y$ matrices.

Consider the following $n-1$ integral solution, each having only two clusters, where the first cluster consists of $\left\{ v_1,\ldots,v_i\right\}$ and the second contains $\left\{ v_{i+1},\ldots,v_n\right\}$ (here $i=1,\ldots,n-1$).
Fixing $i$ and using the above translation of an integral solution to a feasible solution for the semi-definite relaxation, we assign each $v_j$, where $j=1,\ldots,i$ to $\be _1$ and each $v_j$, where $j=i+1,\ldots,n$, to $\be _2$.
Let $Y^i$ be the resulting (positive semi-definite) inner product matrix.
Additionally, consider one additional integral solution that consists of a single cluster containing all of $V$.
In this case, the above translation yields that all $v_i$ vectors are assigned to $\be _1$.
Denote by $Y^n$ the resulting (positive semi-definite) inner product matrix.
Clearly, each of the $Y^1,\ldots,Y^n$ defines a feasible solution for the above natural semi-definite relaxation.

Our fractional solution is given by the average of all the above inner product matrices: $\overline{Y}\triangleq \frac{1}{n}\sum _{i=1}^n Y^i$.
Obviously, $\overline{Y}$ defines a feasible solution for the above natural semi-definite relaxation.
Note that $\by _{v_1} \cdot \by _{v_n} = \frac{n-1}{n}\cdot 0 + \frac{1}{n}\cdot 1 = \frac{1}{n}$ and that $\by _{v_i} \cdot \by _{v_{i+1}}=\frac{n-1}{n}\cdot 1 + \frac{1}{n}\cdot 0 = \frac{n-1}{n}$, for every $i=1,\ldots,n-1$.
Therefore, we can conclude that:
\begin{align*}
A(v_i) & = 2\left(\frac{n-1}{n}\right)=2-\frac{2}{n} &\forall i=2,\ldots,n-1 \\
A(v_1) & =A(v_n) =1-\frac{1}{n}+\left(\frac{n-1}{n}\right)=2-\frac{2}{n}  &
\end{align*}
This demonstrates that the above instance also has an integrality gap of $\nicefrac[]{n}{(2(n-1))}$ for the natural semi-definite relaxation. \hfill $\square$
\end{proof}

\end{document}